\documentclass[11pt]{article}
\pdfoutput=1
\usepackage{amsmath, amssymb, amsthm, graphicx, color, colortbl}

\newcommand{\bitem}{\begin{itemize}}
\newcommand{\eitem}{\end{itemize}}
\newcommand{\benum}{\begin{enumerate}}
\newcommand{\eenum}{\end{enumerate}}
\newcommand{\beqr}{\begin{eqnarray*}}
\newcommand{\eeqr}{\end{eqnarray*}}
\newcommand{\beqrn}{\begin{eqnarray}}
\newcommand{\eeqrn}{\end{eqnarray}}
\newcommand{\barr}{\begin{array}}
\newcommand{\earr}{\end{array}}
\newcommand{\beq}{\begin{equation}}
\newcommand{\eeq}{\end{equation}}
\newcommand{\bfig}{\begin{figure}}
\newcommand{\efig}{\end{figure}}

\newcommand{\iid}{\stackrel{\mbox{\tiny IID}}{\sim}}

\newcommand{\graphsc}[2]{\includegraphics[scale = #1]{#2}}

\newcommand{\rline}{{\mathnormal R}}

\newcommand{\ex}{{\mathrm{E}}}
\newcommand{\var}{{\mathrm{Var}}}
\newcommand{\cov}{{\mathrm{Cov}}}

\newtheorem{thm}{Theorem}[section]

\usepackage[margin = 1.2in]{geometry}
\usepackage{algorithmic, algorithm,natbib,subfigure}

\title{Adaptive Gaussian Predictive Process Approximation}
\author{Surya T Tokdar\\{\small \it Duke University}}
\date{}

\newcommand{\lambdam}{\hat \lambda}
\newcommand{\Lambdam}{\hat\Lambda}
\newcommand{\Psim}{{\hat \Psi}}
\newcommand{\tol}{{\mathsf{tol}}}
\newcommand{\knots}{{knots}}
\newcommand{\data}{{data}}

\begin{document}
\maketitle

\begin{abstract}
We address the issue of knots selection for Gaussian predictive process methodology. Predictive process approximation provides an effective solution to the cubic order computational complexity of Gaussian process models. This approximation crucially depends on a set of points, called knots, at which the original process is retained, while the rest is approximated via a deterministic extrapolation. Knots should be few in number to keep the computational complexity low, but provide a good coverage of the process domain to limit approximation error. We present theoretical calculations to show that coverage must be judged by the canonical metric of the Gaussian process. This necessitates having in place a knots selection algorithm that automatically adapts to the changes in the canonical metric affected by changes in the parameter values controlling the Gaussian process covariance function. We present an algorithm toward this by employing an incomplete Cholesky factorization with pivoting and dynamic stopping. Although these concepts already exist in the literature, our contribution lies in unifying them into a fast algorithm and in using computable error bounds to finesse implementation of the predictive process approximation. The resulting adaptive predictive process offers a substantial automatization of Guassian process model fitting, especially for Bayesian applications where thousands of values of the covariance parameters are to be explored.\\

\end{abstract}

\section{Introduction}

Bayesian nonparametric methodology is driven by construction of prior distributions on function spaces. Toward this, Gaussian process distributions have proved extremely useful due to their mathematical and computational tractability and ability to incorporate a wide range of smoothness assumptions. Gaussian process models have been widely used in spatio-temporal modeling \citep{handcock.stein.93, kim.etal.05, banerjee&etal08}, computer emulation \citep{kennedy&ohagan01, oakley&ohagan02}, non-parametric regression and classification \citep{neal98, csato.etal.00, rasmussen&williams06, short.etal.07}, density and quantile regression \citep{tokdar&etal10, tokdar&kadane11}, functional data analysis \citep{shi&wang08, petrone&etal09}, image analysis \citep{sudderth&jordan09}, etc. \citet{rasmussen&williams06} give a thorough overview of likelihood based exploration of Gaussian process models, including Bayesian treatments. For theoretical details on common Bayesian models based on Gaussian processes, see \citet{tokdar&ghosh07}, \citet{choi&schervish07}, \citet{ghosal&roy06}, \citet{vandervaart&vanzanten08, vandervaart&vanzanten09} and the references therein. 

For our purpose, a Gaussian process can be viewed as a random, real valued function $\omega = (\omega(t), t \in T)$ on a compact Euclidean domain $T$, such that for any finitely many points $t_1, \cdots, t_k \in T$ the random vector $(\omega(t_1), \cdots, \omega(t_k))$ is a $k$-dimensional Gaussian vector. A Gaussian process $\omega$ is completely characterized by its mean and covariance functions $\mu(t) = \ex[ \omega(t)]$ and $\psi(s, t) = \cov[\omega(s), \omega(t)]$. For a Gaussian process model, where a function valued parameter $\omega$ is assigned a Gaussian process prior distributions, the data likelihood typically involves $\omega$ through a vector $W = (\omega(s_1), \cdots, \omega(s_N))$ of $\omega$ values at a finite set of points $s_1, \cdots, s_N \in T$. These points could possibly depend on other model parameters. The fact that $W$ is a Gaussian vector makes computation conceptually straightforward.

However, a well known bottleneck in implementing Gaussian process models is the $O(N^3)$ complexity of inverting or factorizing the $N\times N$ covariance matrix of $W$. Various reduced rank approximations to covariance matrices have been proposed to overcome this problem \citep{smola.bartlett.01, seeger.etal.03, schwaighofer.tresp.03, candela&williams05, snelson&ghahramani06}, mostly reported in the machine learning literature. Among these, a special method of approximation, known as predictive process approximation \citep{tokdar07, banerjee&etal08}, has been independently discovered and successfully used in the Bayesian literature. The appeal of this method lies in a stochastic process representation of the approximation that obtains from tracking $\omega$ at a small number of points, called knots, and extrapolating the rest by using properties of Gaussian process laws (Section \ref{pp}).

For predictive process approximations, choosing the number and locations of the knots remains a difficult problem. This problem is only going to escalate as more complex Gaussian process models are used in hierarchical Bayesian modeling, with rich parametric and non-parametric formulations of the Gaussian process covariance function becoming commonplace. To understand this difficulty, we first lay out (Section \ref{accuracy}) the basic probability theory behind the approximation accuracy of the predictive process and demonstrate how the choice of knots determines an accuracy bound. The key concept here is that the knots must provide a good coverage of the domain $T$ of the Gaussian process. While this is intuitive, what needs emphasis is that the geometry of $T$ is to be viewed through the topology induced by the Gaussian process canonical metric $\rho(s, t) = [\ex\{\omega(s) - \omega(t)\}^2]^{1/2}$, which could be quite different from the Euclidean geometry of $T$. 

This theory helps understand (Section \ref{why}) why existing approaches of choosing knots, based on space filling design concepts \citep{zhu.stein.05, zimmerman.06, finley&etal09} or model extensions where knots are learned from data as additional model parameters \citep{snelson&ghahramani06, tokdar07, guhaniyogi.etal.11}, are likely to offer poor approximation and face severe computational difficulties. A fundamental weakness of these approaches is their inability to automatically adapt to the changes in the geometry of $T$ caused by changes in the values of the covariance parameters.

In Section \ref{app} we present a simple extension of the predictive process approximation that enables it to automatically adapt to the geometry of the Gaussian process covariance function. This extension, called {\it adaptive predictive process approximation}, works with an equivalent representation of the predictive process through reduced rank Cholesky factorization (Section \ref{chol}) and adds to it two adaptive features, {\it pivoting} and {\it dynamic stopping}. Pivoting determines the order in which knots are selected from an initial set of candidate points while dynamic stopping determines how many knots to select. The resulting approximation meets a pre-specified accuracy bound (Section \ref{pds}). 

The connection between predictive process approximation and reduced rank Cholesky factorization is well known \citep{candela&williams05} and pivoting has been recently investigated in this context from the point of view of stable computation \citep{foster&etal09}. The novelty of our work lies in unifying these ideas to define an adaptive version of the predictive process and in proposing accuracy bounds as the driving force in finessing the implementation of such approximation techniques. The end product is a substantial automatization of fitting Gaussian process models that can broaden up the scope of such models without the additional burden of having to model or learn the knots. This is illustrated with two examples in Section \ref{applic}.

\section{Predictive process approximation}

\subsection{Definition}
\label{pp}
Fix a set of points, referred to as {\it knots} hereafter, $\{t_1, t_2,\cdots, t_m\} \subset T$ and write $\omega(t) = \nu(t) + \xi(t)$, where, 
\[
\nu(t) = E\{\omega(t) \;|\; \omega(t_1),\cdots,\omega(t_m)\}
\]
and $\xi(t) = \omega(t) - E\{\omega(t) | \omega(t_1),\cdots,\omega(t_m)\}$. By the properties of Gaussian process laws, $\nu$ and $\xi$ are independent Gaussian processes. The process $\nu$, called a Gaussian predictive process, has rank $m$, because it can be written as $\nu(t) = \sum_{i = 1}^m A_i \psi(t_i,t)$ with the coefficient vector $A = (A_1, \cdots, A_m)$ being a Gaussian vector. By replacing $\omega$ with $\nu$ in the statistical model, one now deals with the covariance matrix of $V = (\nu(s_1), \cdots, \nu(s_N))$, which, due to the rank-$m$ property of $\nu$, can be factorized in $O(Nm^2)$ time. 

\subsection{Accuracy bound}
\label{accuracy}

Replacing $\omega$ with $\nu$ can be justified as follows. Let $\delta = \sup_{t\in T} \min_{1\le i\le m} \rho(t, t_i)$ denote the mesh size of the knots, where $\rho(t, s) = [E\{\omega(t) - \omega(s)\}^2]^{1/2}$ is the canonical metric on $T$ induced by $\omega$ \citep[][page 2]{adler90}. For a smooth $\psi(t, s)$, $\delta$ can be made arbitrarily small by packing $T$ with sufficiently many, well placed knots. But, as $\delta$ tends to $0$, so does $\kappa^2 = \sup_{t\in T} \var \{\xi(t)\}$. This is because for any $t \in T$, and any $i \in \{1, \cdots, m\}$, by the independence of $\nu$ and $\xi$,
\begin{align*}
\var\{ \xi(t)\} & = \var\{ \omega(t)\} - \var\{ \nu(t)\}\\
& = \ex[ \var\{\omega(t) | \omega(t_1), \cdots, \omega(t_m)\}] \\
& = \ex[ \var \{ \omega(t) - \omega(t_i) | \omega(t_1), \cdots, \omega(t_m)\}] \\
& \le \var\{ \omega(t) - \omega(t_i)\} = \rho^2(t, t_i),
\end{align*}
and hence $\kappa \le \delta$. That $\kappa$ can be made arbitrarily small is good news, because it plays a key role in providing probabilistic bounds on the residual process $\xi$.
\begin{thm}
\label{thm bound}
Let $\omega$ be a zero mean Gaussian process on a compact subset $T \subset \rline^p$. Let $\nu$ be a finite rank predictive process approximation of $\omega$ with residual process $\xi = \omega - \nu$.
\benum
\item[(i)] If $T \subset [a,b]^p$ and there is a finite constant $c > 0$ such that $\var\{\omega(s) - \omega(t)\} \le c^2 \|s - t\|^2$, $s, t \in T$ then 
\beq
P\left(\sup_{t \in T} |\xi(t)| > \epsilon\right) \le 3  \exp\left(-\frac{\epsilon^2 }{B^2\kappa}\right), \;\; \forall \epsilon > 0
\label{tail}
\eeq
with $B = 27\sqrt{2pc(b - a)}$  and $\kappa^2 = \sup_{t \in T} \var \{\xi(t)\}$.
\item[(ii)] For any finite subset $S \subset T$
\beq
P\left(\sup_{t \in S} |\xi(t)| > \epsilon\right) \le 3  \exp\left\{-\frac{\epsilon^2 }{9\kappa_S^2(2 + \log |S|)}\right\}, \;\; \forall \epsilon > 0
\label{tail2}
\eeq
where $|S|$ denotes the size of $S$  and $\kappa_S^2 = \sup_{t \in S} \var \{\xi(t)\}$.
\eenum
\end{thm}

A proof is given in Appendix \ref{appndx}. Note that the constant $B$ does not depend on the number or locations of the knots, it only depends on the dimensionality and size of $T$ as well as smoothness properties of the covariance function $\omega$. It is possible to replace $\kappa$ in (\ref{tail}) with $\kappa^{2(1  - \eta)}$ for any arbitrary $\eta \in (0, 1)$, but with a different constant $B$. While (\ref{tail}) provides an accuracy bound over the entire domain $T$, the bound in (\ref{tail2}) over a finite subset maybe of more practical value. 

For Gaussian process regression models with additive Gaussian noise, a common modification \citep{finley&etal09} of predictive process approximation is to replace $\omega$ with the process $\tilde \nu = \nu + \xi^*$ where $\xi^*$ is a zero mean Gaussian process, independent of $\nu$ and $\xi$, satisfying,
\[
\cov \{ \xi^*(t),\xi^*(s)\} = \left\{\begin{array}{ll} \cov\{\xi(t), \xi(s)\} = \var \xi(t) & \mbox{if } t = s\\
0 & \mbox{if } t \ne s.
\end{array}
\right.
\]
The addition of $\xi^*$ gives $\tilde \nu$ the same pointwise mean and variance as those of $\omega$, without adding to the computational cost. The residual process is now given by $\tilde \xi = \omega - \tilde \nu = \xi - \xi^*$ whose variance equals $2\var\{ \xi(t)\}$ because of independence between $\xi$ and $\xi^*$. Because $\xi^*$ is almost surely discontinuous, the bound in (\ref{tail}) does not apply to $\tilde \xi$. But (\ref{tail2}) continues to hold with $\kappa_S^2$ replaced by $\tilde\kappa_S^2 = 2\kappa_S^2$.


%

 

\subsection{Need for adaptation}
\label{why}
For predictive process approximations, choosing the number and the locations of the knots remains a difficult problem. Ideally this choice should adapt to the canonical metric $\rho$, so that a small $\delta$ obtains with as few knots as possible. However, it's not a single $\rho$ that we need to adapt to. In modern Gaussian process applications, the covariance function $\psi$ and consequently the canonical metric depend on additional model parameters. A typical example is $\omega$ of the form
\[
\omega(t) = \omega_0(t) + x_1(t)\omega_1(t) + É + x_p(t)\omega_p(t)
\]
where $x_j(t)$'s are known, fixed functions and $\omega_j$'s are independent mean zero Gaussian processes with covariances $\psi_j(t, s) = \tau_j^2\exp(-\beta_j^2 \|s - t\|^2)$ with $\theta = (\tau_0,\tau_1, \cdots, \tau_p, \beta_0, \beta_1, \cdots, \beta_p)$ serving as a model parameter. Because a likelihood based model fitting will loop through hundreds or even thousands of values of $\theta$, it is important to have a low-cost algorithm to choose the knots that automatically adapts to the geometry of any arbitrary canonical metric.

Such an adaptive feature is lacking from existing knots selection approaches which primarily treat knots as additional model parameters. The knots are then learned along with other model parameters, either via optimization \citep{snelson&ghahramani06} or by Markov chain sampling \citep{tokdar07, guhaniyogi.etal.11}. Another popular approach is to work with a fixed set of knots based on space-filling design concepts \citep{zhu.stein.05, zimmerman.06, finley&etal09}. Among these, the proposal in \citet{finley&etal09} overlaps with our proposal. But while we pursue a low-cost adaptation at every value of $\theta$ at which likelihood evaluation is needed, \citet{finley&etal09} consider one fixed set of knots adapted to a representative value of $\theta$. Their knot selection algorithm has $O(N^2m)$ computing time, which makes it infeasible to run within an optimization or a Markov chain sampler loop.

\section{Adaptative predictive process}
\label{app}
\subsection{Predictive process approximation via  Cholesky factorization}
\label{chol}

Let $\omega = (\omega(t), t \in T)$ be a zero-mean Gaussian process with covariance $\psi(s, t)$. Suppose the finite set $S = \{s_1, s_2, \cdots, s_N\} \subset T$ contains all points in $T$ where $\omega$ needs to be evaluated for the purpose of model fitting. Let $\Psi = ((\psi_{ij}))$ denote the $N\times N$ covariance matrix of the Gaussian vector $W = (\omega(s_1), \cdots, \omega(s_N))$. A Cholesky factor $\Lambda$ of $\Psi$, with $\Lambda$ being a $N \times N$ upper triangular matrix with non-negative diagonal elements and satisfying $\Psi = \Lambda'\Lambda$, obtains from the following recursive calculations
\beq
\label{recursion}
\lambda_{ii} = \sqrt{\psi_{ii} - \sum_{\ell < i} \lambda_{\ell i}^2},\;\; \lambda_{ij} = \frac{\psi_{ij} - \sum_{\ell < i} \psi_{\ell j}\psi_{\ell i}}{\lambda_{ii}},\;\; i = 1, \cdots, N, j = i + 1, \cdots, N
\eeq
with $\lambda_{ij}$, $j < i$ set to zero. This gives a row-by-row construction of $\Lambda$ and requires $C_M Ni^2$ computation time for constructing the first $i$ rows for some machine dependent constant $C_M$.

For an $m \in \{1, \cdots, N\}$, an approximation $\Lambdam = ((\lambdam_{ij}))$ to $\Lambda$ obtains in $O(Nm^2)$ time by an incomplete application of the above recursion. The first $m$ rows of $\Lambdam$ are constructed in $C_MNm^2$ time through (\ref{recursion}):
\beq
\label{incomp recursion}
\lambdam_{ii}  =  \sqrt{\psi_{ii} - \sum_{\ell < i} \{\lambdam_{\ell i}\}^2},\;\; \lambdam_{ij} = \frac{\psi_{ij} - \sum_{\ell < i} \psi_{\ell j}\psi_{\ell i}}{\lambdam_{ii}},\;\; i = 1, \cdots, m, j = i + 1, \cdots, N. 
\eeq
For the remaining rows, only the diagonal elements are computed in $C_M(N - m)$ time as in the first part of (\ref{recursion}), with the off diagonals set to zero:
\beq
\label{diag aug}
\lambdam_{ii}  =  \sqrt{\psi_{ii} - \sum_{\ell \le m} \{\lambdam_{\ell i}\}^2},\;\; \lambdam_{ij} = 0,\;\;i = m + 1, \cdots, N, j = i + 1, \cdots, N. 
\eeq
The lower triangular elements $\lambdam_{ij}$, $j < i$, are all set to 0. The resulting $\Lambdam$ is upper triangular with non-negative diagonals and equals the Cholesky factor of the covariance matrix $\Psim$ of $V = (\nu(s_1), \cdots, \nu(s_N))$ where $\nu = (\nu(t), t \in T)$ is the Gaussian predictive process approximation of $\omega$ based on knots $s_1, \cdots, s_m$. The resulting residual process $\xi = \omega - \nu$ satisfies $\var \{\xi(s_i)\} = 0$, $1 \le i \le m$ and $\var \{\xi(s_i)\} = \{\lambdam_{ii}\}^2$, $i = m + 1, \cdots, N$, and hence $\kappa_S^2 := \max_{s \in S} \var \{\xi(s)\} = \max_{i > m} \lambdam_{ii}^2$. From (\ref{tail2}), $\kappa_S^2$ controls error bounds $P(\max_{s \in S} |\omega(s) - \nu(s)| > \epsilon)$ over the set of interest $S$.

Therefore, the above incomplete Cholesky factorization produces a Gaussian predictive process approximation, with readily available error bounds,  provided we are happy to choose the knots from the set $S$. The restriction to $S$ appears reasonable for most applications with the additional burden on the modeler to identify $S$ carefully. For example, in a Gaussian process regression model with additive noise, it is sufficient to take $S$ to be the training set of covariate values, if only posterior predictive mean and variances are needed at test cases. But if posterior predictive covariance between two test cases, or a test and a training case is required, then $S$ should include these test cases as well.

\subsection{Pivoting and dynamic stopping}
\label{pds}
Our quest of an adaptive version of $\nu$ stays within this restriction, but employs a dynamic choice of the stopping time $m$ and the order in which the elements of $S$ are processed.  To decide whether the current stopping time is acceptable, we check the current $\kappa_S$ against a given tolerance $\kappa_\tol$. If $\kappa_S$ exceeds the tolerance, we increment $m$ to $m + 1$, and repeat (\ref{incomp recursion}) only for $i$ equal to the new value of $m$, followed by (\ref{diag aug}), producing an update of $\kappa_S$. The top row elements $\lambdam_{ij}$, $i < m$ need no changes.

The increment of $m$ and the subsequent alterations to $\Lambdam$ clearly reduce $\kappa_S$ as the tailing $\lambdam_{ii}$'s in (\ref{diag aug}) are reduced. This reduction can be expected to improve if after incrementing $m$ and before proceeding with the new calculations, one swaps the current $m$-th and $k$-th rows of $\Psi$ where $i = k$ gives the maximum of the tailing $\lambdam_{ii}$ values, $i = m, \cdots, N$. A sequence of such swaps, from start to the terminating $m$, gives a greedy, dynamic approximation to finding the optimal ordering of the elements of $S$ that gives a $\kappa_S \le \kappa_\tol$ with a minimum stopping time $m$. 

The dynamic swapping is a common feature, known as {\it pivoting}, of all leading software packages for Cholesky factorization. If run until $m = N$, pivoting produces a permutation $\pi = (\pi_1, \cdots, \pi_N)$ of $(1, \cdots, N)$ and an upper triangular matrix $\Lambda$ with non-negative diagonals such that $P_\pi \Psi P_{\pi}' = \Lambda'\Lambda$ where $P_\pi$ is the $N\times N$ permutation matrix associated with $\pi$. Our proposal above simply adds to this pivoted Cholesky factorization a dynamic, tolerance based stopping. The resulting $\Lambdam$ gives the Cholesky factor of the covariance matrix of $V = (\nu(s_1), \cdots, \nu(s_N))$ where $\nu$ is the Gaussian predictive process associated with the knots $s_{\pi_1}, \cdots, s_{\pi_m}$. The Gaussian predictive process $\nu$ comes with the error bound (\ref{tail2}) with $\kappa_\tol$ replacing $\kappa$. The additional computing time needed for pivoting is only $O(Nm)$, a small fraction of the computing time $O(Nm^2)$ needed to get the elements $\Lambdam$ if $\pi$ was precomputed. 

Algorithm \ref{A:spchol} gives a pseudo code for performing the incomplete Cholesky factorization with an additional improvisation. The user specified tolerance $\kappa_\tol$ is taken to be a relative tolerance level instead of an absolute one. The absolute tolerance is fixed as $\kappa_\tol$ times the maximum of $\var\{\omega(s_i)\}^{1/2}$, $i = 1, \cdots, N$. This makes sense for Gaussian process approximation as the maximum variance of the process can be viewed as a scaling parameter.

\begin{algorithm}
\caption{Pivoted, incomplete Cholesky factorization with dynamic stopping.}
\label{A:spchol}
\begin{algorithmic}
\REQUIRE {A covariance function $\psi(\cdot, \cdot)$, positive integers $N$ and $m_{\max}$ and tolerance $\kappa_\tol > 0$.}
\STATE $R \gets 0_{m_{\max} \times m_{\max}}$
\STATE $\pi_1 \gets 1, \pi_2 \gets 2, \cdots, \pi_N \gets N$
\STATE $k \gets 1$
\STATE $d_{\max} \gets \max_{1 \le l \le N} \psi(s_{\pi_l}, s_{\pi_l})$
\STATE $l_{\max} \gets \arg\max_{1 \le l \le N} \psi(s_{\pi_l}, s_{\pi_l})$
\STATE $\kappa_\tol \gets \sqrt{d_{\max}}\kappa_\tol$
\WHILE{$d_{\max} > \kappa_\tol^2$}
\STATE {\bf swap} $\pi_k$ and $\pi({l_{\max}})$
\STATE $r_{kk} \gets d_{\max}^{1/2}$
\FOR{$j = k + 1$ to $m$}
\STATE $r_{kj} \gets \{\psi(s_{\pi_k}, s_{\pi_j}) - \sum_{l < k} r_{lk} r_{lj}\} / r_{kk}$
\ENDFOR
\STATE $k \gets k + 1$
\STATE $d_{\max} \gets \max_{k \le l \le N} \psi(s_{\pi_l}, s_{\pi_l}) - \sum_{l < k} r_{lk}^2$
\STATE $l_{\max} \gets \arg\max_{k \le l \le N} \psi(s_{\pi_l}, s_{\pi_l}) - \sum_{l < k} r_{lk}^2$\\
\ENDWHILE
\STATE $m \gets k$
\FOR{$k = m + 1$ to $N$}
\STATE $r_{kk} \gets \{\psi(s_{\pi_k}, s_{\pi_k}) - \sum_{l \le m} r_{lm}^2\}^{1/2}$
\ENDFOR
\RETURN Factor matrix $R$, pivot $\pi$ and rank $m$
\end{algorithmic}
\end{algorithm}

\subsection{Geometric Illustration}

We illustrate the adaptive choice of knots on the Bartlett experimental forest dataset \citep{finley&etal09}. This dataset contain $n = 437$ well identified forest locations $s_1, \cdots, s_n$, measured as eastings (on the horizontal axis) and northings (vertical axis) from a reference point (Figure \ref{fig 1}). For illustration purposes, we consider only a hypothetical Gaussian process model where a surface $\omega(s)$ over the forest area is modeled as a zero-mean Gaussian process with a square-exponential covariance function
\beq \psi(s, t) = x(s) x(t) \exp(-\beta\|Q (s - t)\|^2) \label{bef cov},\eeq
for some constant $\beta > 0$, an orthogonal projection matrix $Q$ and some fixed function $x(t)$ that relates to the slope of the forest landscape at location $t$. The variance of $\omega(t)$ is $x(t)^2$ and the correlation between $\omega(s)$ and $\omega(t)$ depends on the Euclidean distance between the $Q$-projections of these two location vectors. The parameter $\beta > 0$ encodes the spatial range of the covariance, i.e., it controls how rapidly $\psi(s, t)$ decays with the distance between $s$ and $t$. We demonstrate how the choice of knots according to Algorithm \ref{A:spchol} adapts to variations in each of $\beta$, $x(t)$ and $Q$. 

\bfig
\centering
\subfigure[\label{1a}]{\graphsc{0.75}{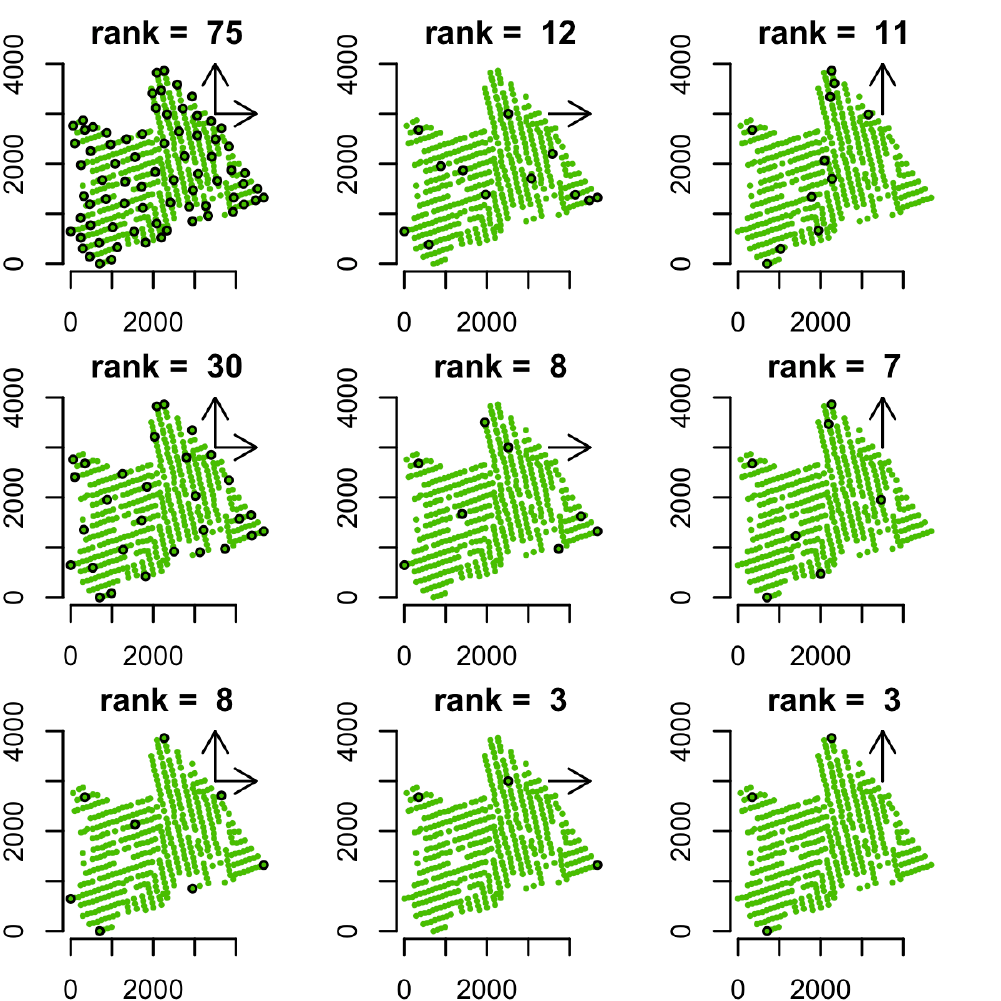}} 
\subfigure[\label{1b}]{\graphsc{0.75}{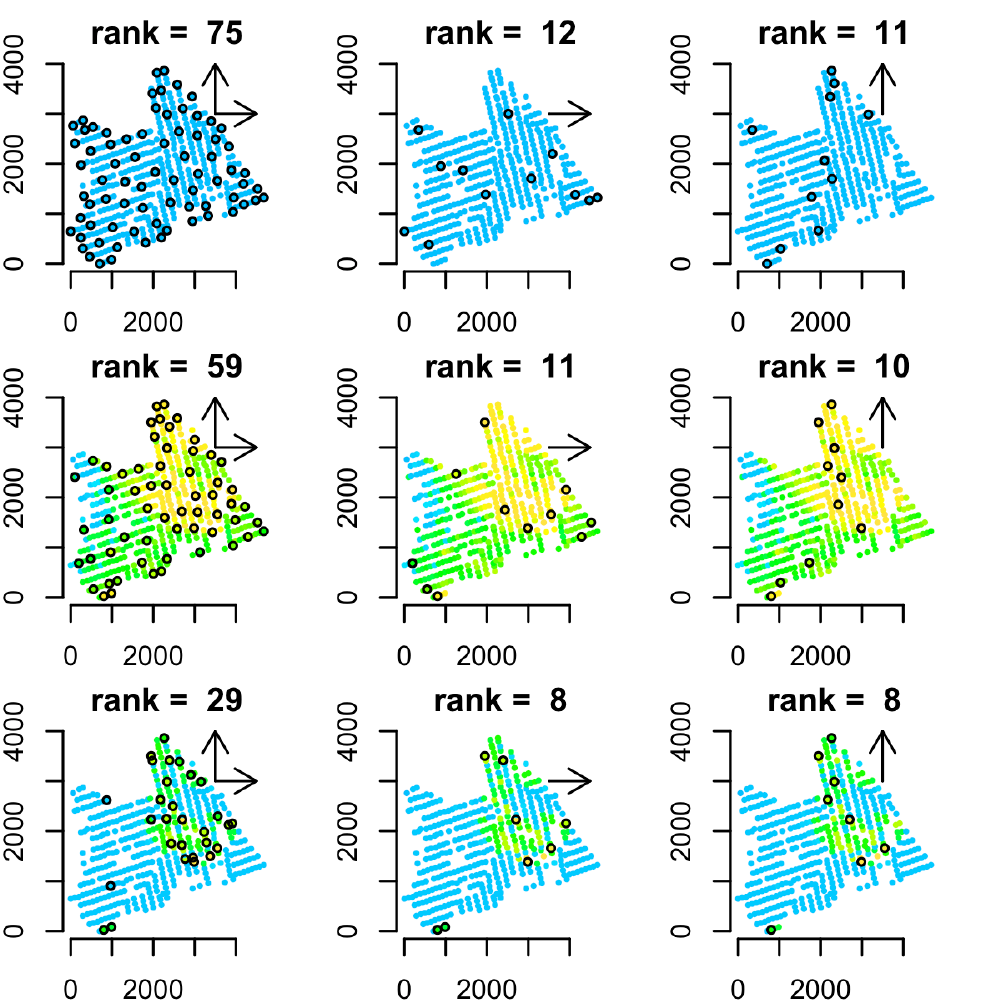}} 
\subfigure[\label{1c}]{\graphsc{0.75}{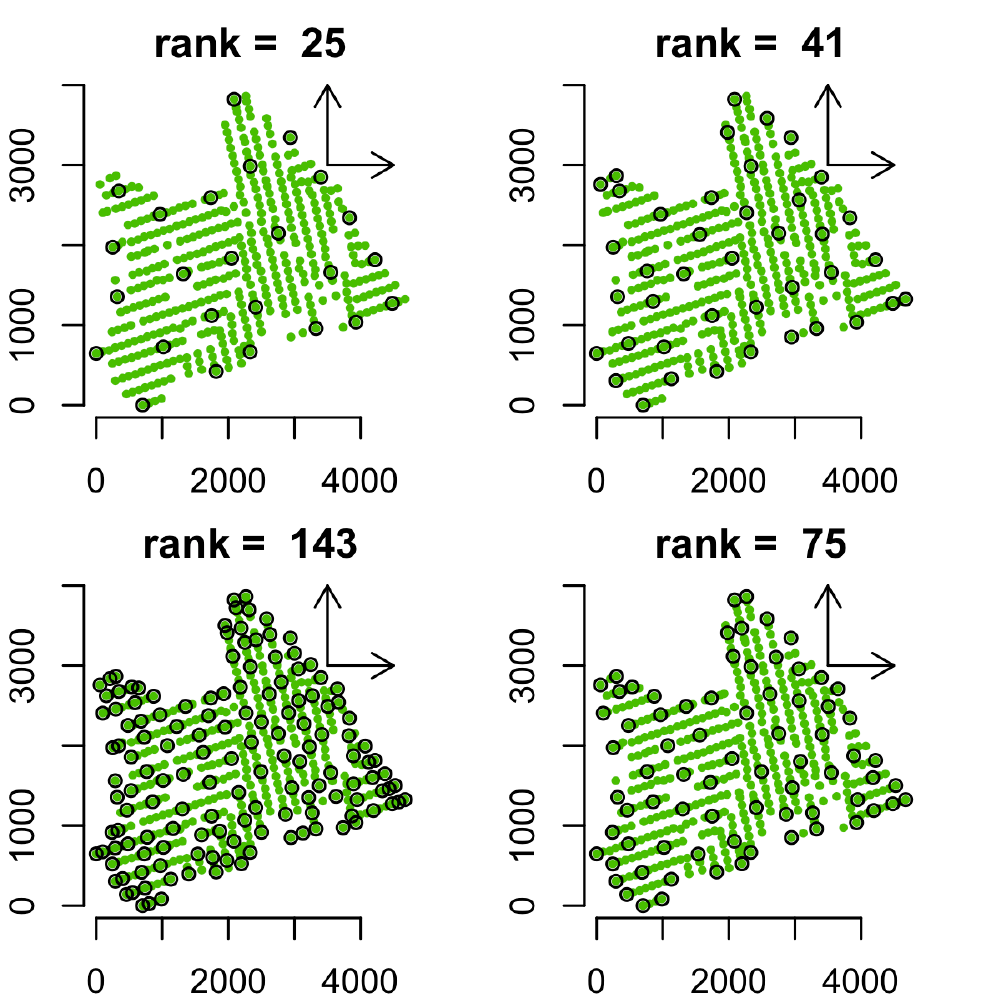}} 
\subfigure[\label{1d}]{\graphsc{0.75}{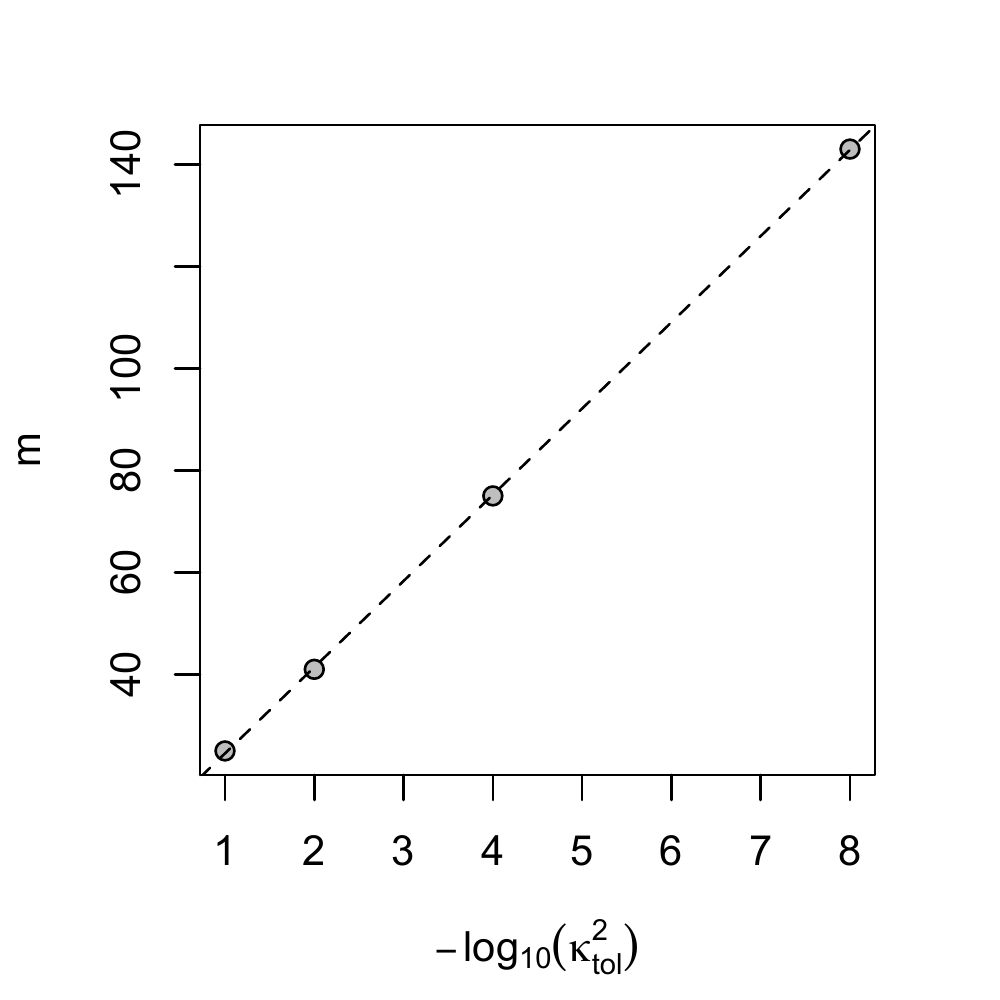}}
\caption{Geometry of knot selection illustrated on Bartlett experimental forest data. (a) Adaptation of knots to changes in $Q$ (columns) and $\beta$ (rows). (b) The same for changes in $Q$ (columns) and $x(s)$ (rows). (c) The same for changes in $\kappa_\tol$. (d) Number of knots needed to meet specified accuracy bounds for a particular choice of the covariance.}
\label{fig 1}
\efig

Figure \ref{1a} shows nine choices of (\ref{bef cov}) that vary in $Q$ and $\beta$, while $x(t)$ and $\kappa_\tol$ are held fixed at $x(t) \equiv 1$ and $\kappa_\tol^2 = 10^{-4}$. The top row has $\beta = 10^{-3}$, the middle row has $\beta = 5\times 10^{-4}$ and the bottom row has $\beta = 10^{-4}$. The left column has $Q$ equal to the identity matrix, the middle column has $Q$ that projects along the horizontal axis and the right column has $Q$ that projects along the vertical axis (indicated by arrows). It is clear that the algorithm picks fewer knots as the spatial range decreases. A smaller spatial range gives a flatter topology in the canonical metric, consequently fewer points are required to capture the variation in the random surface $\omega(t)$. It is also clear that the algorithm adapts to the directional element of this topology. It lines up the knots along the horizontal axis when $Q$ is the horizontal projection and lines up the knots along the vertical axis when $Q$ projects along that axis.

Figure \ref{1b} shows nine other choices of (\ref{bef cov}) that vary in $Q$ and $x(t)$, while $\beta$ and $\kappa_\tol$ are held fixed at $\beta = 10^{-3}$ and $\kappa_\tol^2 = 10^{-4}$. Variation in $Q$ is as in Figure \ref{1a}. We take $x(t) = 1 - F(c \cdot \texttt{slope}(t) | a, b)$, where $F(x|a, b)$ denotes the gamma distribution function with shape parameter $a$ and rate parameter $b$, and $\texttt{slope}(t)$ is the slope of the landscape at point $t$. We fix $a$ and $b$ so that the mean and variance of the gamma distribution match the mean and variance of the recorded slope values at the 437 locations. The top row of Figure \ref{1b} has $c = 0$, the middle row has $c = 1$ and the bottom row has $c = 6$. Larger values of $c$ make $x(t)$ closer to zero at regions with a high slope while $x(t)$ always equals 1 at regions that are flat (shown in lighter color, mostly along a narrow valley running from south-east to north-west). 
With a larger $c$, most of the variation in $\omega(t)$ is confined to locations $t$ with a flat slope. It is clear that our algorithm adapts to this feature by picking knots from such areas.

Figure \ref{1c} shows (\ref{bef cov}) with $\beta = 10^{-4}$, $x(t) \equiv 1$, $Q =$ the identity matrix, but with four choices of $\kappa_\tol^2 = 10^{-1}, 10^{-2}, 10^{-4}$ and $10^{-8}$ (clockwise from top left). For smaller tolerance levels, more knots are picked to meet a tighter accuracy condition. A good coverage of the entire region is maintained throughout, but additional knots are chosen to give a denser representation. Note the higher concentration of knots at the boundary than the interior. This is a consequence of the greedy nature of the algorithm as it tries to pick the next knot as the point that is least correlated with the ones already selected. Although a better algorithm could correct for such a boundary bias, the linear additional computing cost of the greedy search offers a highly attractive trade-off against a few extra knots. Figure \ref{1d} indicate that the number of knots $m$ required to meet the tolerance criterion grows at a logarithmic rate with $1 / \kappa_\tol^2$. However, theoretical results are not yet available on the relationship between $m$ and $\kappa_\tol$.

\section{Application to Bayesian computation}
\label{applic}
\subsection{Sparse nonparametric regression}

A low cost, adaptive choice of knots is extremely beneficial in Bayesian Markov chain Monte Carlo \citep{mcmcbook1, mcmcbook2} computations, where likelihood evaluation is needed at thousands of different values of the covariance parameters. We illustrate this with a sparse non-parametric regression model, where a good exploration of the space of covariance parameters is critical to model fitting. We show that to efficiently explore the covariance parameter space, it is important to adapt to the changes in the canonical metric caused by the changes in these parameter values. In this regard, the adaptive predictive process approximation proposed here has a clear advantage over the existing approaches of handling knots.



We consider a toy data set consisting of $(x_i, y_i)$, $i = 1, \cdots, n = 10,000$, where $x_i = (x_{i1}, \cdots, x_{ip})'$ are drawn independently from the uniform distribution over $[0,1]^p$, with $p = 10$, and $y_i$ are generated independently as $y_i \sim N(2\sin \{2\pi x_{i1}\}, 0.1^2)$. We consider a regression model
\beq
y_i = \mu + \tau \omega(x_i) + \tau\epsilon_i, \;\;\epsilon_i \iid N(0, \sigma^2),
\eeq
with $\omega$ modeled as a zero-mean Gaussian process over $T = [0,1]^p$. The covariance function of $\omega$ is taken to be
$$\psi(s, t) = \psi(s, t | \beta) = \exp\{-\sum_{j = 1}^p \beta_j^2 (s_j - t_j)^2\},\;\; t, s \in T.$$
For simplicity of exposition, we set $\mu$ and $\tau^2$ at the mean and variance of the observed $y_i$ values, and focus on learning the parameters $\beta_1, \cdots, \beta_p$ and $\sigma^2$. The $\beta_j$ parameters are assigned standard normal prior distributions, folded onto the positive real line and $\log \sigma^2$ is assigned a standard normal prior distribution. Prior independence across parameters is assumed.

A fixed set of knots over $T$ is clearly infeasible for this application due to the dimensionality of $T$. Placing only 3 knots along each axis takes the total count to $3^{10} = 59049$. While the alternative approach of placing an auxiliary model on the knots and learning them jointly with the covariance parameters via Markov chain sampling can drastically reduce the total number of knots, it is likely to lead to a poor exploration of the covariance parameters for the following reason.

\bfig
\centering
\graphsc{0.75}{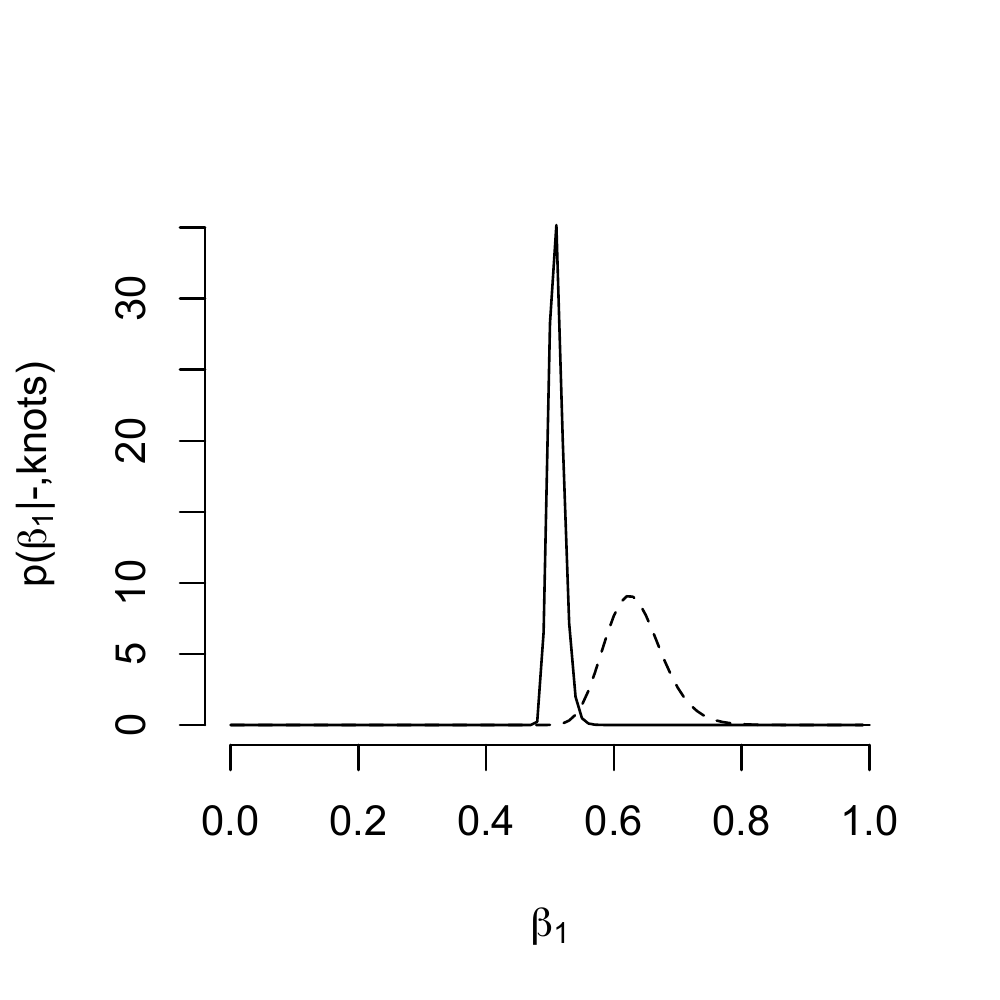}
\caption{Conditional posterior density of $\beta_1$, in the non-parametric regression model, given other parameters and the knots, for two different choices of the knots.}
\label{fig1.5}
\efig

Consider a Gibbs update for $\beta_1$ given the remaining parameters and the knots. Suppose the current parameter values are $\beta_1 = \beta_2 = 0.1$, $\beta_3 = \cdots = \beta_p = 0.004$ and $\sigma^2 = \exp(-3.5)$. For these parameter values, a ``well learned'' choice of the knots is found by applying Algorithm \ref{A:spchol} to $\psi(s, t |\beta)$ with $\beta$ set at the vector of current values (we use $\kappa_\tol^2 = 10^{-4}$). The corresponding posterior conditional density $p(\beta_1 | \beta_2, \cdots, \beta_p, \sigma^2, \knots, \data)$ is shown by the solid curve in Figure \ref{fig1.5}. If, instead, the current values had been $\beta_1 = 1$, $\beta_2 = 0.1$, $\beta_3 = \cdots = \beta_p = 0.004$, $\sigma^2 = \exp(-3.5)$ and the knots were chosen by applying Algorithm \ref{A:spchol} to the corresponding $\psi(s, t | \beta)$, then the resulting posterior conditional density $p(\beta_1 | \beta_2, \cdots, \beta_p, \sigma^2, \knots, \data)$ would look like the dashed curve in Figure \ref{fig1.5}. The two curves are very different even though the values of the conditioning parameters $\beta_2, \cdots, \beta_p$ and $\sigma^2$ are the same. This difference shows that the conditional distribution of the covariance parameters strongly depends on the current set of knots, which, if well learned, should depend on the current values of the covariance parameters. Consequently, there will be additional stickiness in the chain of sampled values of the covariance parameters.

\bfig
\centering
\subfigure[]{\graphsc{0.75}{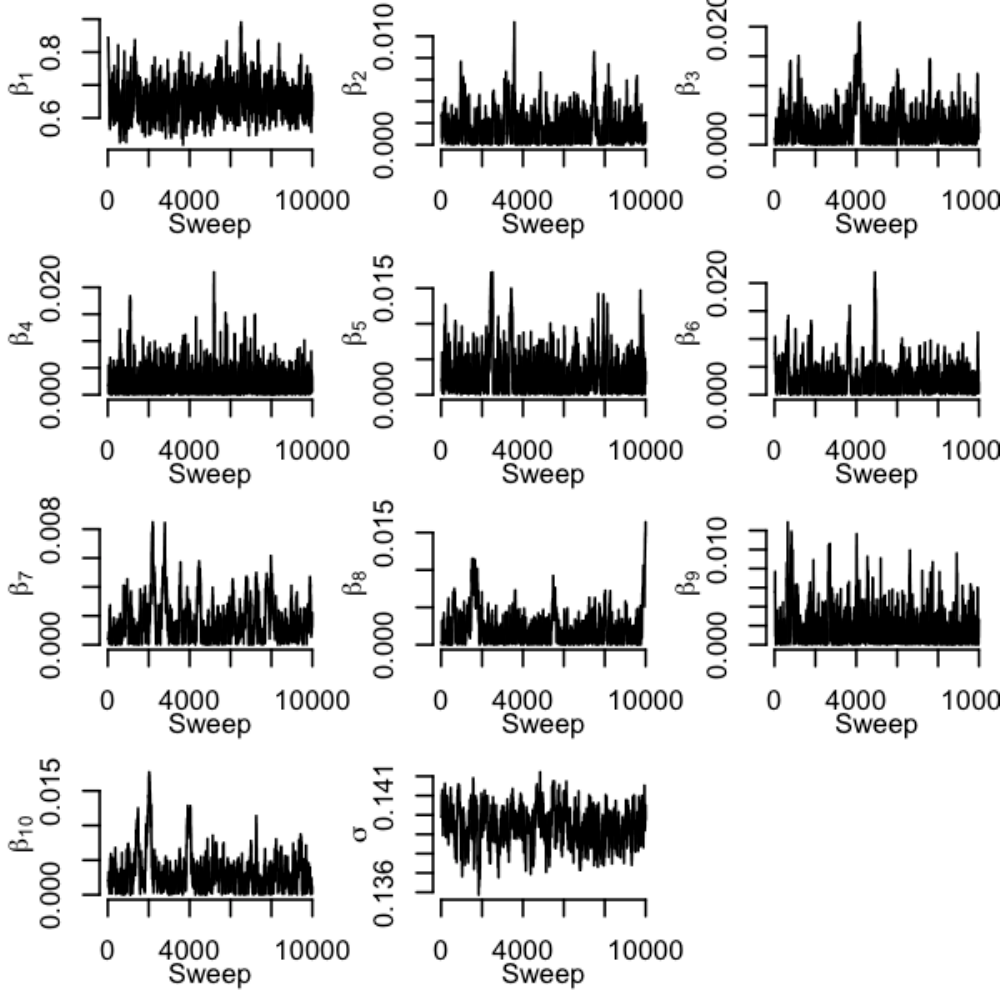}}
\subfigure[]{\graphsc{0.75}{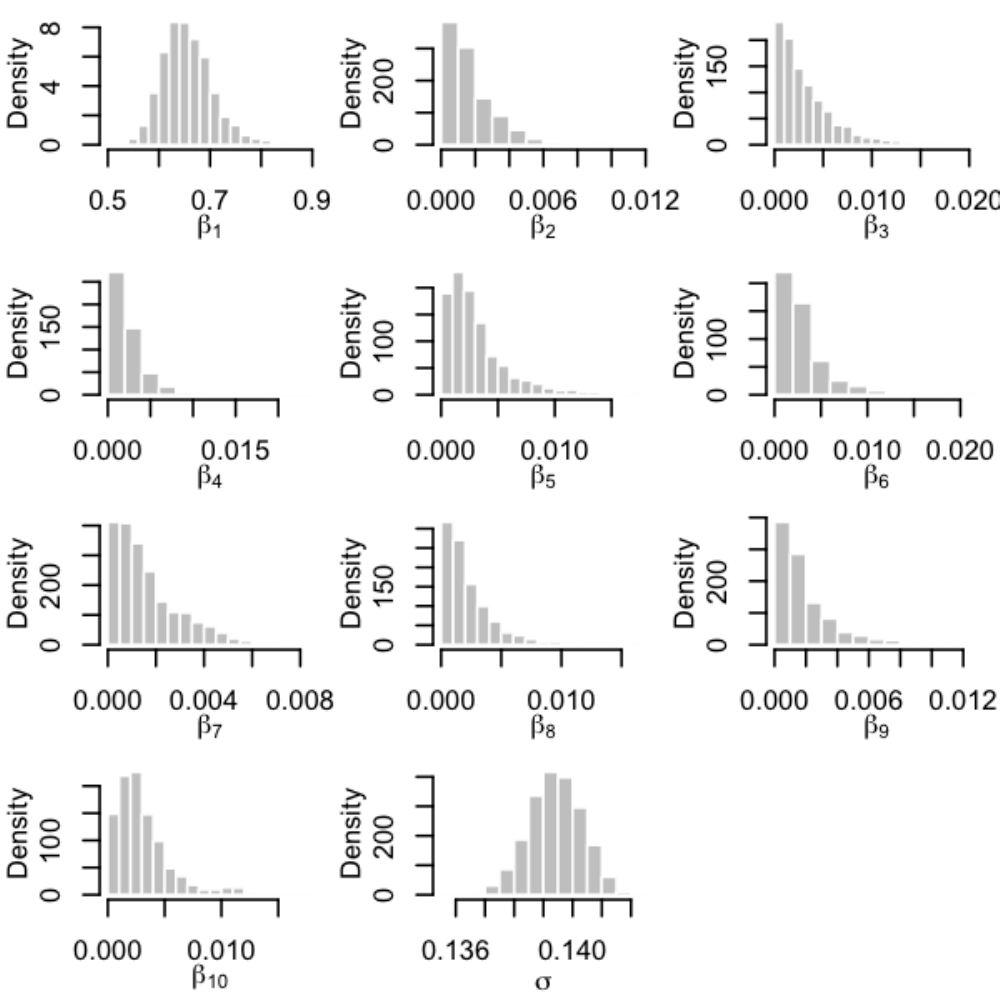}}
\subfigure[]{\graphsc{0.75}{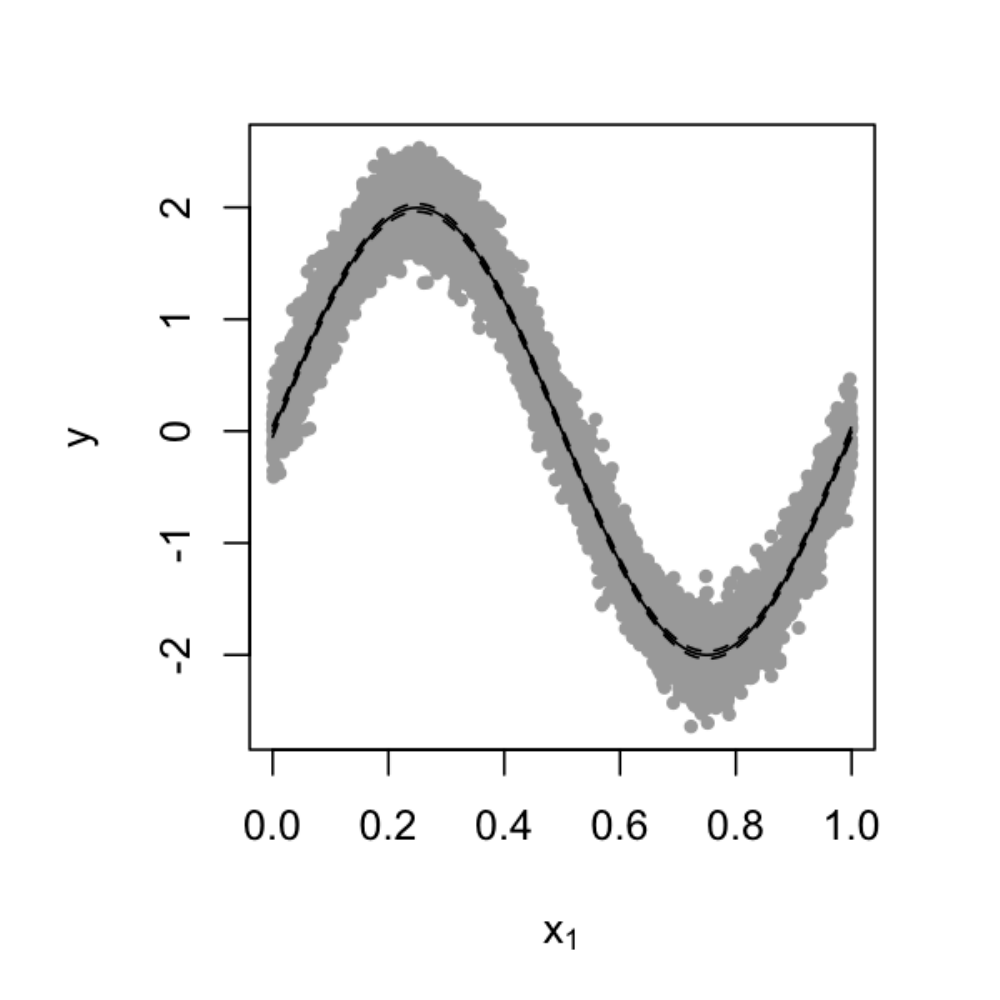}}
\subfigure[]{\graphsc{0.75}{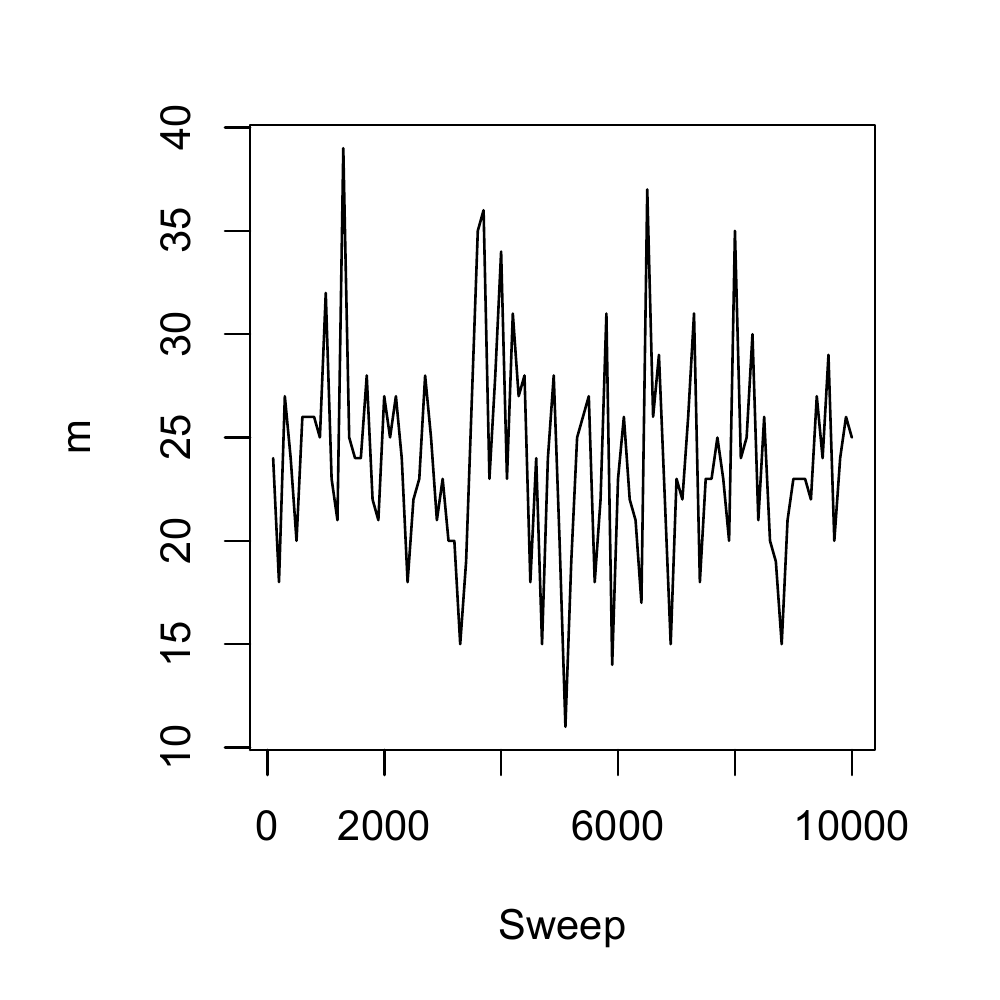}}
\caption{A random walk Metropolis sampler exploration of the posterior for the non-parametric regression problem and associated posterior summaries. (a) Trace plots of the sampled values of $\beta_1, \cdots, \beta_{10}, \sigma$. (b) Histograms of the same. (c) Posterior median (solid line) and 95\% credible intervals (dashed lines) for $f(x)$ at $x = (i/100, 0, \cdots, 0)$, $i = 0, \cdots, 100$, overlaid on the data scatter along $x_1$ and $y$. (d) Values of $m$ along the run of the sampler, shown at every 100th sweep.}
\label{fig 2}
\efig

The proposed adaptive predictive process gets rid of this additional stickiness by automatically adapting the knots to the covariance parameter values. As discussed before, for this application it is reasonable to restrict the search of knots to the set of observed covariate vectors $S = \{x_1, \cdots, x_n\}$. Then, one only needs to specify a tolerance level $\kappa_\tol$ to obtain an approximation $\hat p(y | \beta, \sigma^2)$ of the marginal likelihood $p(y | \beta, \sigma^2) = \int p(y | \omega, \sigma^2) p(d\omega | \beta)$ of $(\beta, \sigma^2)$. The approximation obtains by replacing $\omega$ in the integral with its predictive process approximation $\nu$ adapted to the corresponding $\psi(s, t|\beta)$. The approximate marginal likelihood can be computed in $O(nm^2)$ time, where $m$ is the number of knots needed for $\psi(s,t|\beta)$ with the given tolerance level; detailed formulas are available in \citet{snelson&ghahramani06}. The approximate posterior density $\hat p(\beta, \sigma^2 | \data) \propto \hat p(y | \beta, \sigma^2) p(\beta, \sigma^2)$ can be explored with common Metropolis-Hastings samplers. Figure \ref{fig 2} reports summaries from a random walk Metropolis sampler exploration with a tolerance level $\kappa_\tol^2 = 10^{-4}$.

\subsection{Varying coefficient regression for spatial data}

\citet{pace.barry.97} use county level data from 1980 United States presidential election to relate voter turnout to education, income and homeownership standards. They use a spatial autoregressive model to allow this relation vary geographically. We pursue an alternative formulation with a Gaussian process spatial regression model. 

For county $i = 1, \cdots,n$, let $V_i$, $P_i$, $E_i$, $I_i$ and $H_i$ denote, respectively, voter count, population size of age 18 years or more, population size with at least high school education, aggregate income and number of owner occupied housing units. We take log-percentage voter count $y_i = \log(V_i / P_i)$ as the response variable. Three predictor variables are defined as $x_{1i} = \log (E_i / P_i)$, $x_{2i} = \log (I_i / P_i)$ and $x_{3i} = \log (H_i / P_i)$. 

We relate the response to the predictors through a spatially varying regression model
\beq
y_i = \omega_{0i} + x_{1i}\omega_{1i} + x_{2i}\omega_{2i} + x_{3i}\omega_{3i} + \epsilon_i, \;\;\epsilon_i \iid N(0, \sigma^2).
\eeq 
To ensure that geographically proximate counties have similar coefficients, the vector of coefficients from all counties $(\omega_{j1},\cdots, \omega_{jn})$, for each $j = 0, 1,2, 3$, is taken to be a zero mean multivariate normal with $\cov(\omega_{ji}, \omega_{jk}) = \tau_j^2 \exp\{-\beta_j^2\|t_i - t_k\|^2\}$, where $t_i$ is the spatial location vector of county $i$, given by the latitude-longitude pair of the county centroid. These four vectors are taken to be mutually independent.

The above formulation is equivalent to 
\beq
y_i = \omega(t_i) + \epsilon_i,\;\;\epsilon_i \iid N(0, \sigma^2)
\eeq
where $\omega(t)$ is a zero-mean Gaussian process on $T = \{t_1, \cdots, t_n\}$ with covariance function $\psi(s, t | \theta) = \sum_{j = 0}^3 x_j(t)x_j(s) \exp\{-\beta_j^2\|s - t\|^2\}, s, t \in T$, with $x_0(t_i) \equiv 1$, $x_j(t_i) = x_{ji}$, etc, that depends upon the parameter vector $\theta = (\tau_0, \tau_1, \tau_2, \tau_3, \beta_0, \beta_1, \beta_2, \beta_3)$. These covariance parameters are each assigned a standard normal prior distribution, folded onto the positive half of the real line, independently of each other. We also assign $\log\sigma^2$ a standard normal prior distribution, and $\log\sigma^2$ is taken to be a priori independent of $\theta$.

\bfig
\centering
\subfigure[\label{3a}]{\graphsc{0.75}{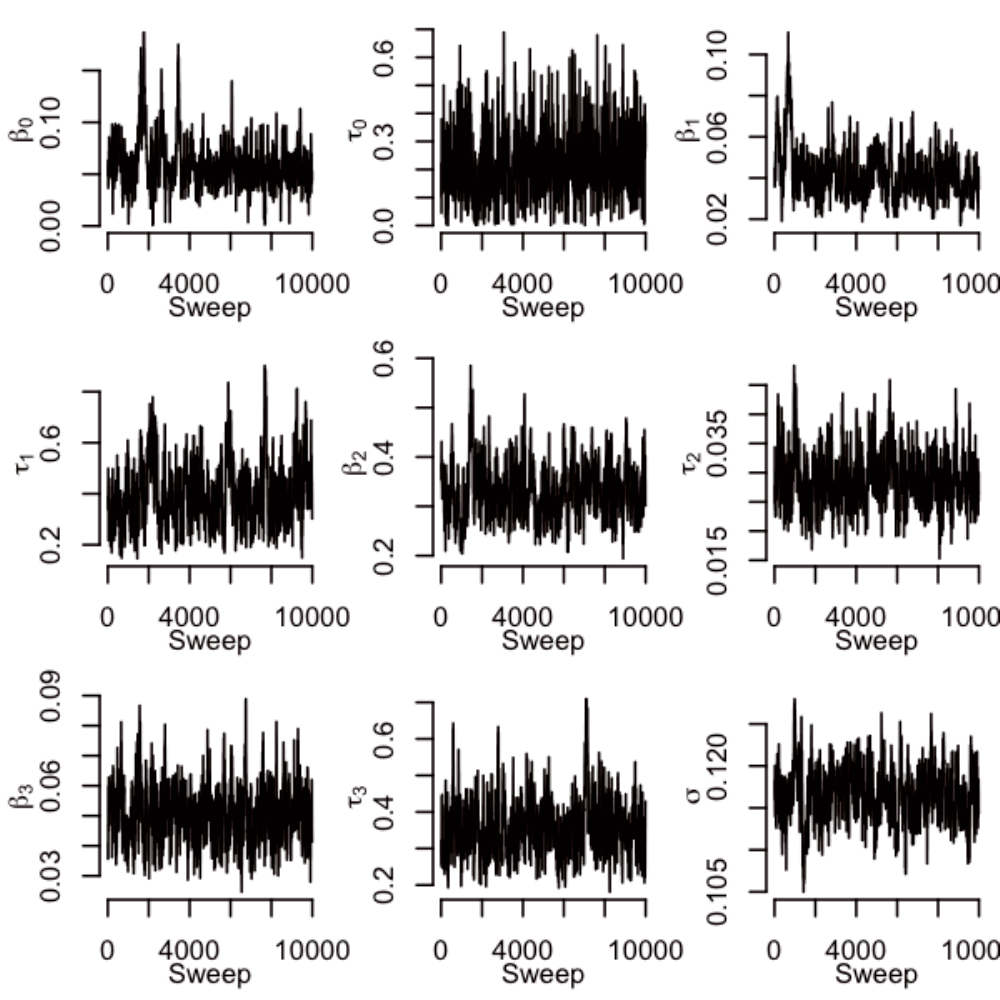}}
\subfigure[\label{3b}]{\graphsc{0.75}{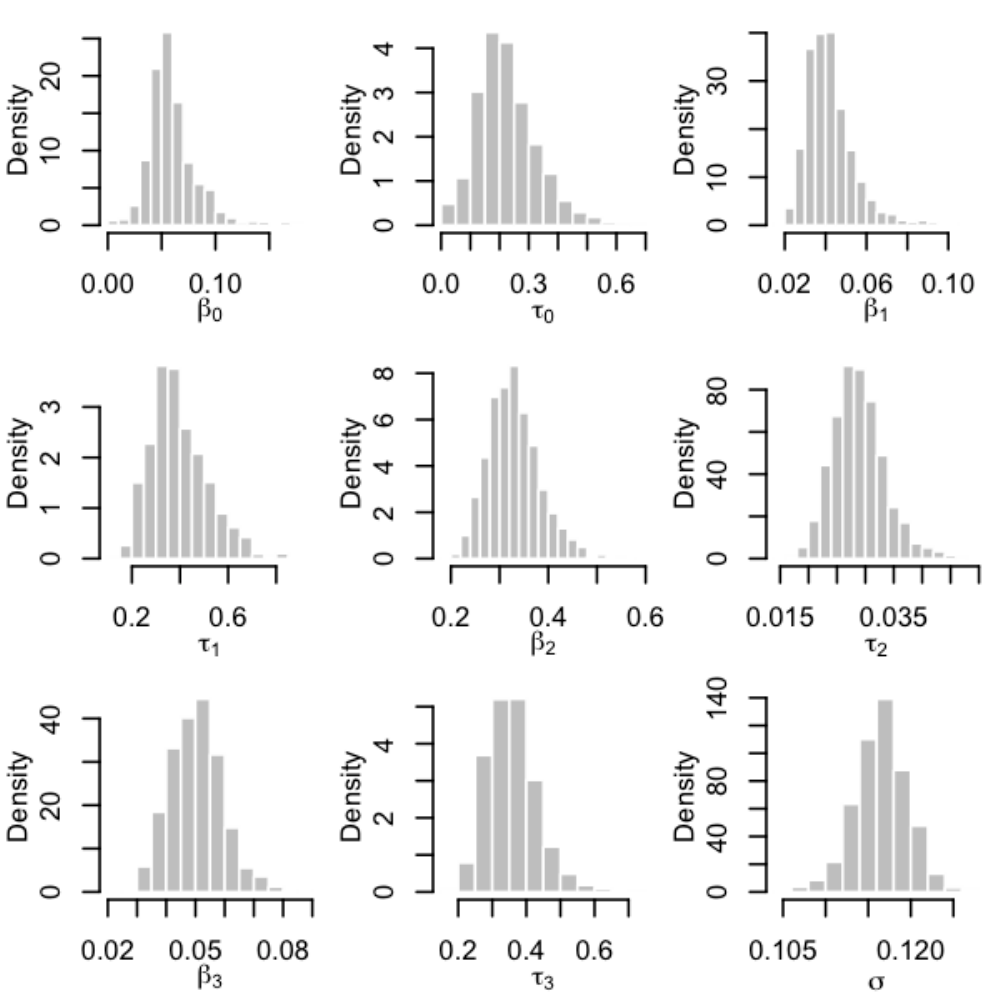}}
\subfigure[\label{3c}]{\graphsc{0.75}{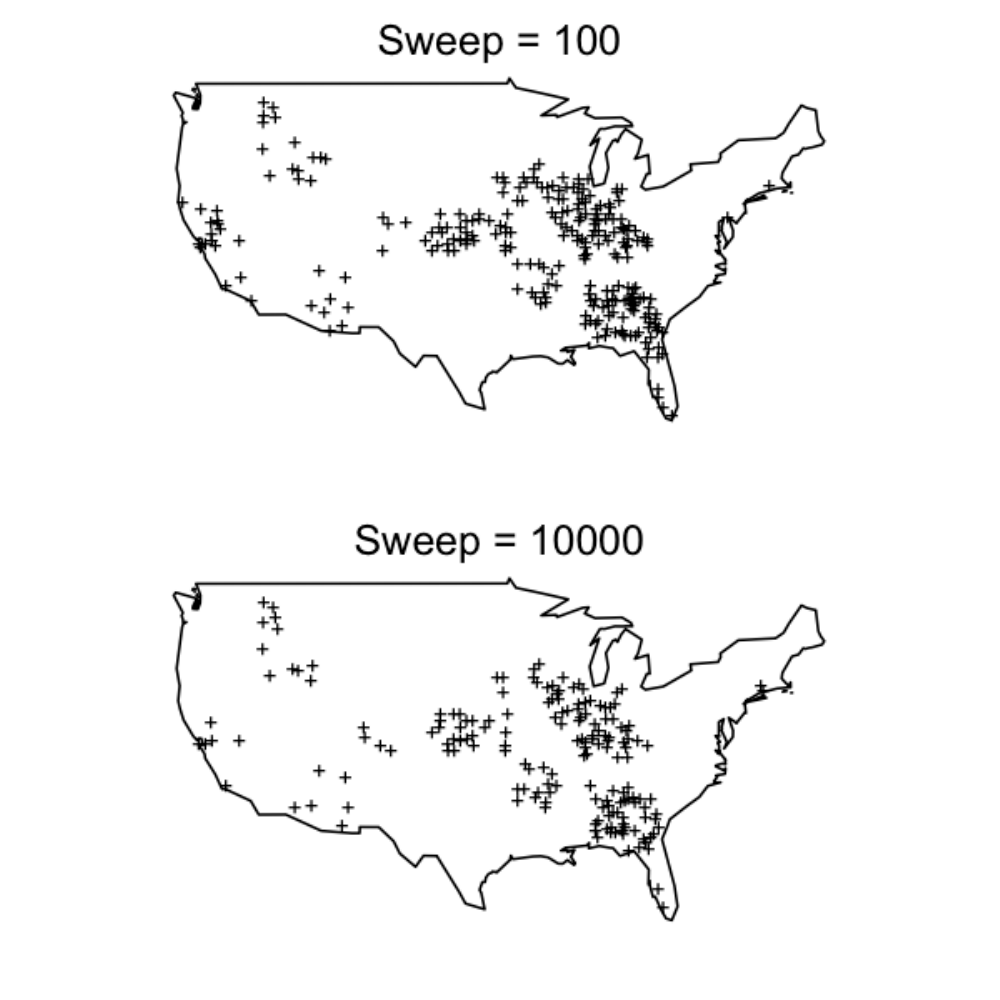}}
\subfigure[\label{3d}]{\graphsc{0.75}{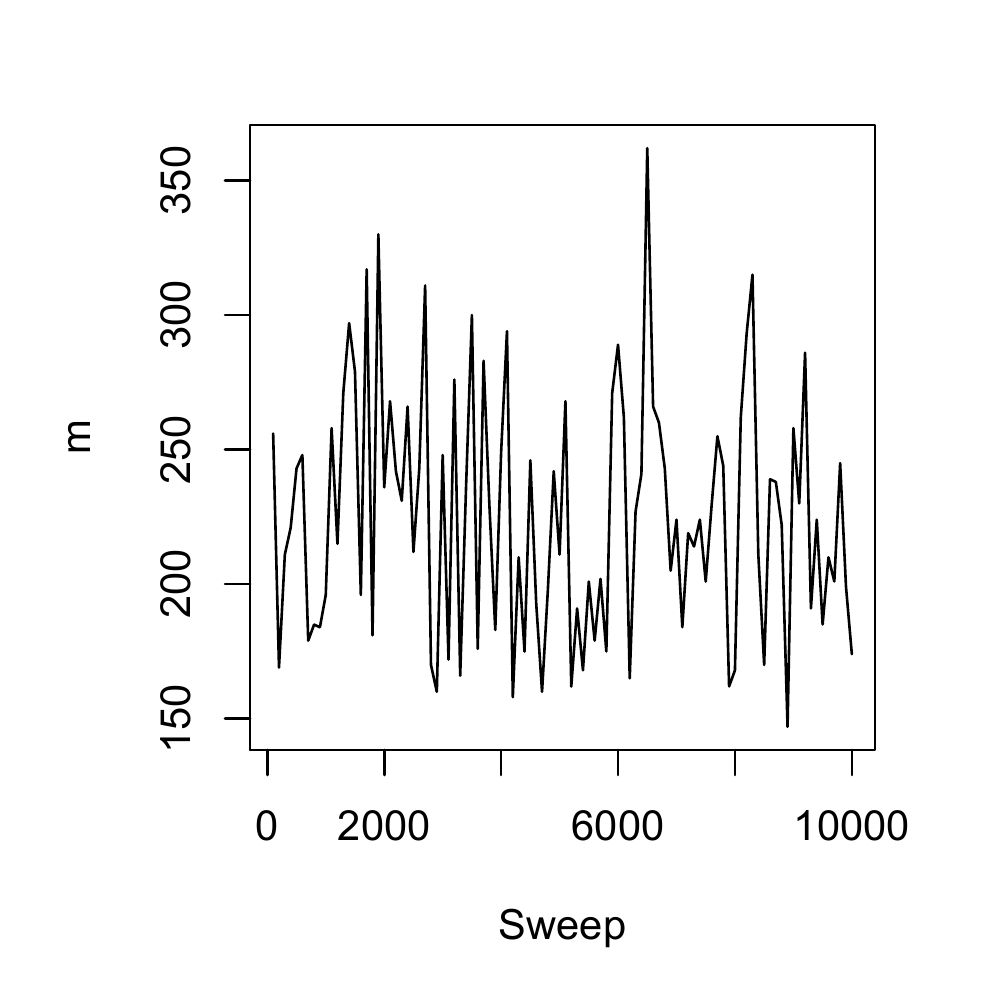}}
\caption{A random walk Metropolis exploration of the posterior for the election turnout analysis. (a) Trace plots of model parameters. (b) Histograms of the same. (c) Knot locations at two distant sweeps of the sampler. (d) Values of $m$ along the sampler.}
\label{fig 3}
\efig

We fit this model using 1040 randomly chosen counties, roughly one third of the total count. Figure \ref{fig 3} shows a random walk Metropolis sampler exploration of the approximate posterior density $\hat p(\theta, \sigma^2 | y) \propto \hat p(y | \theta, \sigma^2) p(\theta)p(\sigma^2)$ with $\kappa_\tol^2 = 10^{-4}$ and $S$ taken to be the set of locations for the counties included in model fitting. Trace plots in Figure \ref{3a} indicate good convergence of the sampler. The marginal posterior distributions of the model parameters appear unimodal (Figure \ref{3b}). Figure \ref{3c} shows knots locations from two different sweeps of the sampler. It is interesting that the knots are not uniformly distributed over the whole country. The sampled values of $m$, shown in Figure \ref{3d} indicate that unlike the regression example of the previous section a substantial fraction ($\sim15\%$-35\%) of points from $S$ are needed to meet the accuracy bound.

\bfig
\centering
\subfigure[\label{4a}]{\graphsc{0.75}{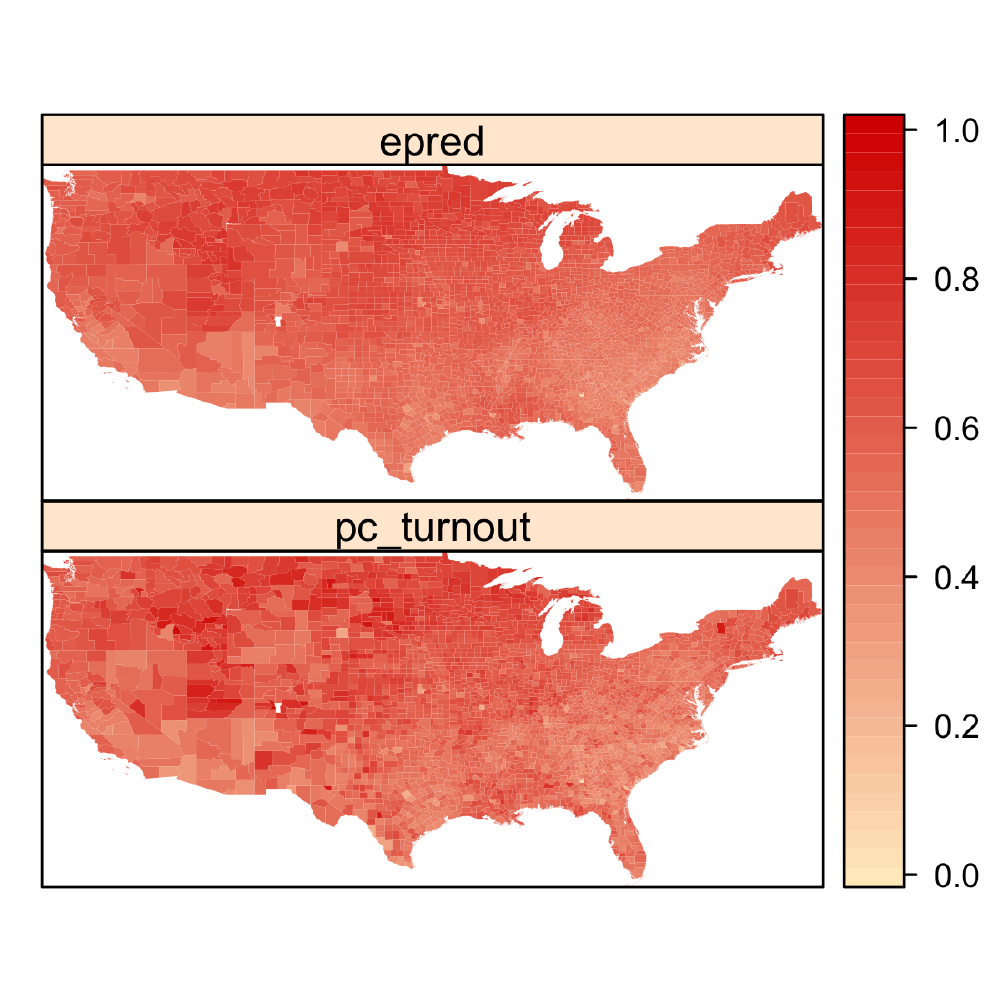}}
\subfigure[\label{4b}]{\graphsc{0.75}{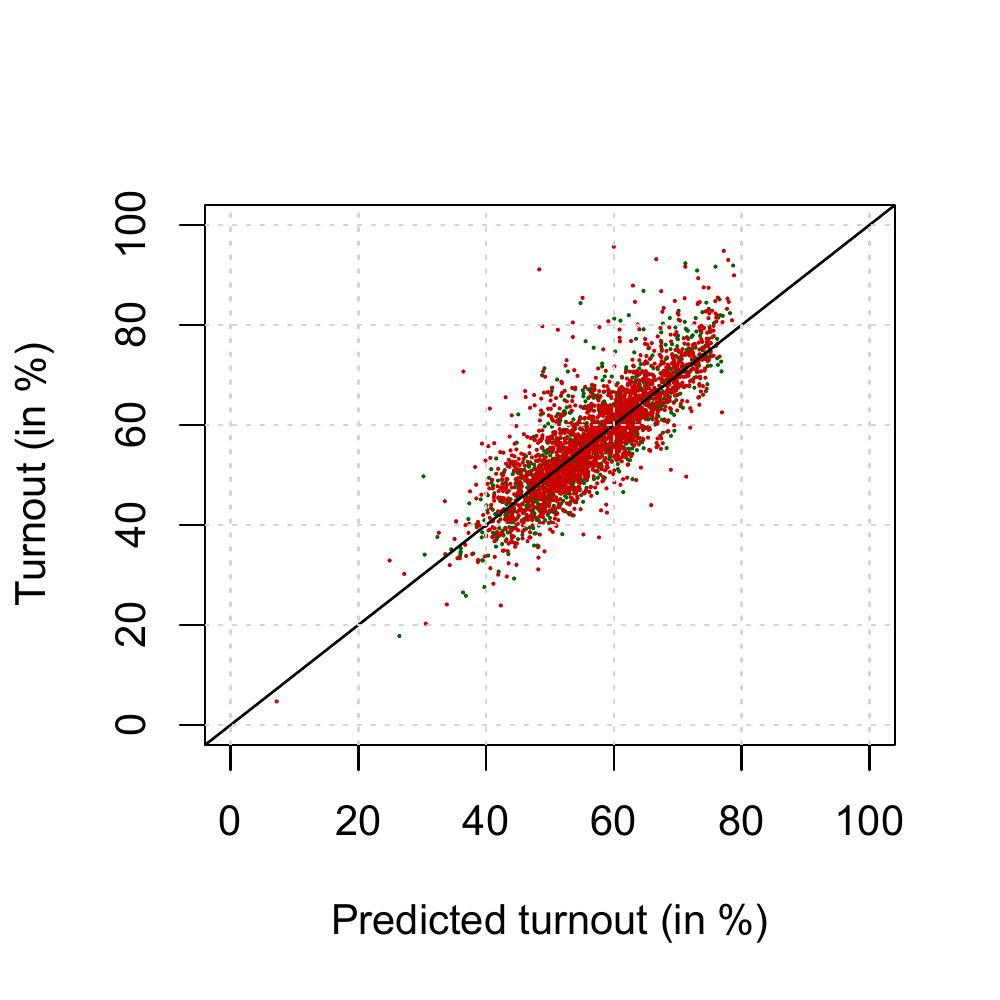}}
\subfigure[\label{4c}]{\graphsc{0.75}{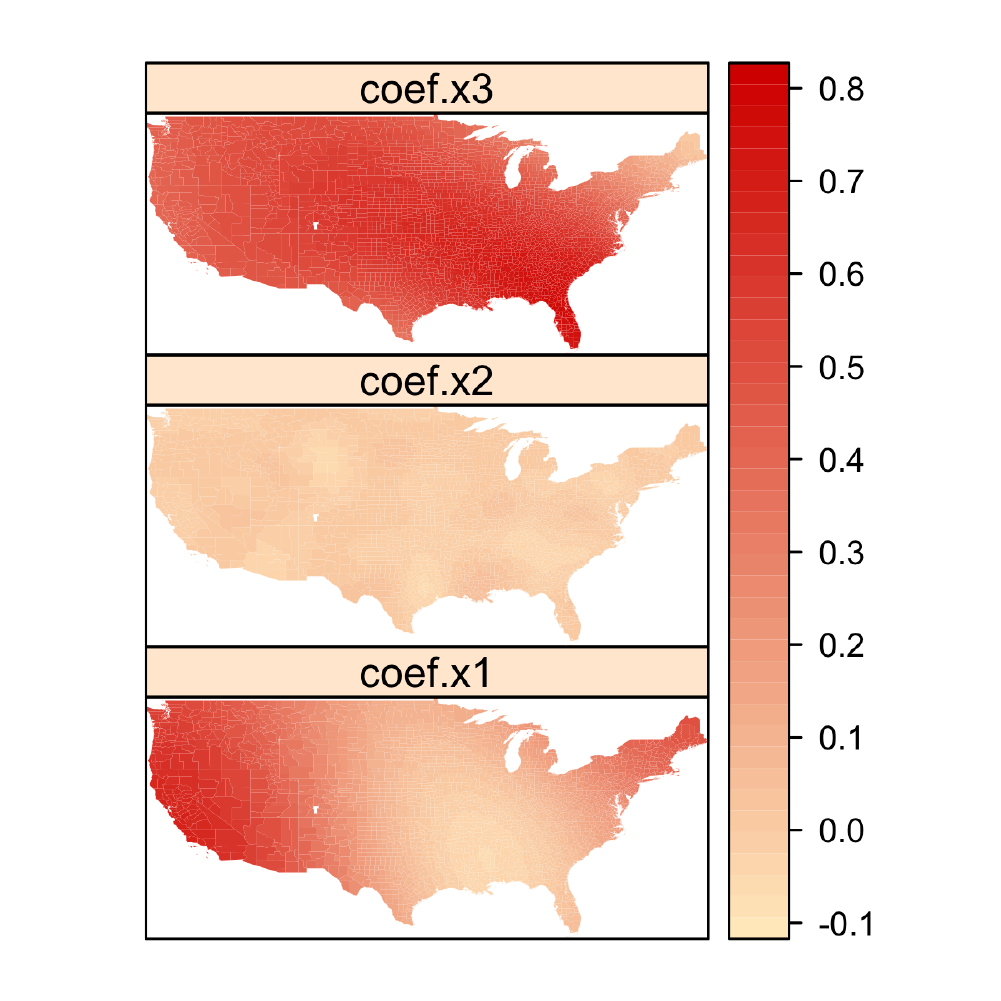}}
\subfigure[\label{4d}]{\graphsc{0.75}{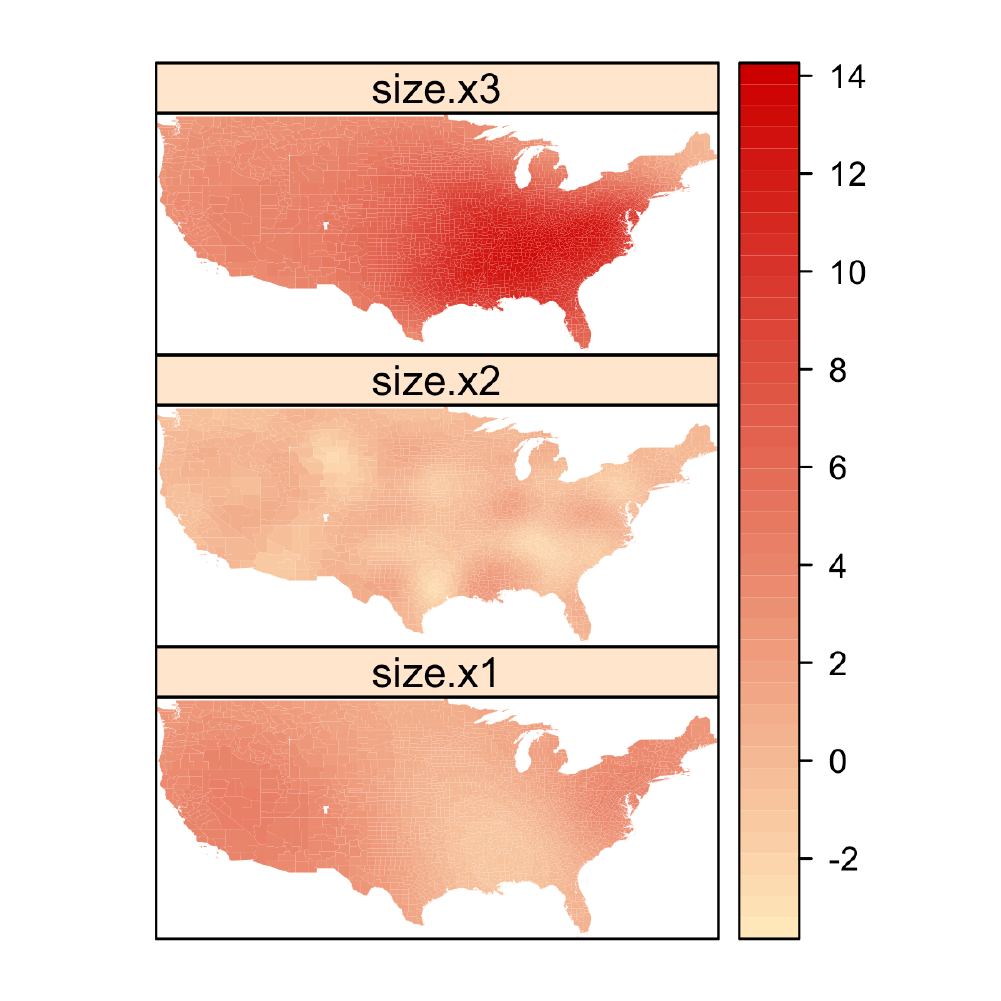}}
\caption{Posterior summary for the election turnout analysis. (a) Observed (bottom) and predicted (top) percentage turnout shown by counties. (b) Scatter plot of the same, green dots in the background mark training data while the red dots in the foreground are for held out data. (c) Estimated coefficients for the three predictors shown by counties. (d) Effect sizes of the three predictors shown by counties.}
\label{fig 4}
\efig

Figure \ref{fig 4} shows summaries of our model fit. Figure \ref{4a} shows the percentage voter turnout (bottom) and the predicted percentage turnout (top) for all 3106 counties. The predicted value for county $i$ is calculated as the Monte Carlo approximation to $\exp[\ex\{\omega(t_i) \mid \data\}]$. Figure \ref{4b} combines the values from Figure \ref{4a} into a scatter plot, split by counties included in model fitting (green dots) and the remaining ones (red dots). Figure \ref{4c} shows the spatially varying regression coefficients $\ex\{\omega_{ji} | \data\}$ on predictors $x_{ji}$, $j = 1,2,3$, for all counties $i$. Figure \ref{4d} shows the effect size of these coefficients, found by $ \ex\{\omega_{ji} | \data\} / \sqrt{\var\{\omega_{ji} | \data\}}$.

From Figure \ref{fig 4} we see that predictor $x_3$, which relates to home ownership, has a strong, spatially varying influence on voter turnout. This predictor has a positive coefficient for all counties, with larger values and effect sizes for counties in the southeast. Predictor $x_1$, which relates to education, has a moderate, spatially varying influence, with about 1/3 of the counties having an effect size larger than 2. It is interesting that coefficients of $x_1$ and $x_3$ appear to be inversely related to each other. Predictor $x_2$, which relates to income, has a weak influence, with effect size less than 2 in absolute value for more than 95\% of the counties. However, the spatial variation of the coefficient of $x_2$ is more intricate and wavy than that of the other two predictors. Although we do not try to interpret these variations, we note that spatial regression model indeed provides a better fit than an ordinary linear regression that relates $y$ to $x_1$, $x_2$ and $x_3$. The root mean square prediction error from the ordinary linear regression, calculated over the counties not included in model fitting, equals 0.14. The same statistic for the spatial regression model is 0.11. 


\section{Discussion}

We have addressed the question of knot selection within the predictive process methodology and have offered proposals that can substantially automatize the implementation of Gaussian process models in Bayesian analysis. A key conceptual contribution of our work is the emphasis on error bounds to derive a finite rank approximation of an infinite dimensional Gaussian process. It must be noted that the accuracy bounds we provide are all {\it a priori}. That is, we can not provide an accuracy bound on how well the posterior distribution under the predictive process model approximates the posterior distribution under the original Gaussian process model. However, \citet{tokdar07} provides some useful theoretical calculation toward this.

We note that approaches that use an auxiliary model on knots are fundamentally different from our deterministic choice of knots driven by accuracy bounds. Our approach clearly stays within the limits of an approximating method. The smaller the specified tolerance, the closer we are to the original Gaussian process. Approaches with a model on the knots essentially define a different stochastic process. The new stochastic process could indeed provide a better model for the given task, as discussed in \citet{snelson&ghahramani06}. However, it would be erroneous to assume that the theoretical properties of a Gaussian process model would also apply to this new stochastic process model. 

It is important to realize that at the crux of a predictive process approximation is the rank deficiency of the covariance matrix of the Gaussian vector $W = (\omega(s_1), \cdots, \omega(s_N))$, which is a manifestation of the underlying smoothness of the process $\omega$. For Gaussian process models that employ a relatively un-smooth $\omega$, such as spatial autoregressive models for lattice data \citep{besag74, rue.etal.09}, predictive process may not be the ideal approximating tool. Indeed, in such cases, the covariance matrix of $W$ need not be rank deficient, but can have special banded structures that allow for other approximation techniques to yield $O(Nm^2)$ computing.

For a smooth $\omega$, the rank deficiency of the covariance matrix of $W$ poses an additional problem. The covariance matrix, irrespective of its size, can be ill conditioned, making numerical computations unstable. The dynamic, tolerance based stopping of the adaptive predictive process can solve this problem to a large extent, because the knot finding algorithm is likely to terminate before the covariance matrix of $(\omega(s_{\pi_1}), \cdots, \omega(s_{\pi_m}))$ becomes ill conditioned. 

\appendix
\section{Technical details}
\label{appndx}

\begin{proof}[Proof of Theorem \ref{thm bound}]
Let $\Psi(x) = \frac{1}{3} e^{x^2}$, $x \ge 0$. $\Psi(x)$ is increasing and convex with $\Psi(0) \in (0, 1)$. For any random variable $Z$ its $\Psi$-Orlicz norm \citep[page 3]{pollard.mono} is defined as $\|Z\|_\Psi := \inf\{C > 0: \ex \Psi(|Z| / C) \le 1\}$. Such a norm provides bounds on tail probabilities as follows: $P(|Z| > x) \le 1 / \Psi(x / \|Z\|_\Psi) = 3\exp(-x^2 / \|Z\|_\Psi^2)$ for all $x > 0$. This immediately leads to (\ref{tail}) and (\ref{tail2}) once we show $\|\sup_{t \in T} |\xi(t)|\|_\Psi \le B\kappa^{1/2}$ and $\|\sup_{t \in S} |\xi(t)\|_\Psi \le 3\kappa_S \sqrt{\log(2 + \log |S|)}$.

\benum
\item[(i)]
Lemma 3.4 of \citet{pollard.mono} states that for any $t_0 \in T$,
\begin{align}
\|\sup_{t \in T} |\xi(t)|\|_\Psi & \le \|\xi(t_0)\|_\Psi + \sum_{i = 0}^\infty \frac{\Delta}{2^i} \sqrt{2 + \log D(\frac{\Delta}{2^i}, T, d)} \label{entropy1}\\
& \le \|\xi(t_0)\|_\Psi + 9\int_0^\Delta \sqrt{\log D(\epsilon, T, d)} d\epsilon \label{entropy2}
\end{align}
where $D(\epsilon, T, d)$ is the $\epsilon$-packing number of $T$ under a (pseudo) metric $d(s, t)$ with $\Delta = \sup_{s,t} d(s,t)$, provided $\|\xi(s) - \xi(t)\|_\Psi \le d(s, t)$ for all $s,t \in T$. It is easy to calculate that $\|Z\|_\Psi = 1.5$ if $Z \sim N(0, 1)$ and that $\|Z\|_\Psi = 1.5 \sigma$ if $Z \sim N(0, \sigma^2)$. Therefore, for any $s, t \in T$, $\|\xi(s) - \xi(t)\|_\Psi  = 1.5 [\var\{\xi(s) - \xi(t)\}]^{1/2}$. Therefore (\ref{entropy1}) holds if we take $d(s,t) = 1.5 [\var \{\xi(s) - \xi(t)\}]^{1/2}$. 

To calculate the right hand side of (\ref{entropy1}), fix $t_0$ to be any of the knots used in defining $\nu$. Then $\xi(t_0) = 0$ and consequently the first term $\|\xi(t_0)\|_\Psi = 0$. To calculate the integral in the second term, furst note that  
$$d(s, t)^2 = 2.25\var\{\xi(s) - \xi(t)\} \le 2.25\var\{\omega(s) - \omega(t)\} \le 2.25 c^2\|s - t\|^2$$
where the first inequality holds because $\omega = \nu + \xi$ with $\nu$ and $\xi$ independent, and the second inequality follows from our assumption on $\omega$. Therefore $D(\epsilon, T, d) \le \{1 + 1.5 c(b - a) / \epsilon\}^d$. Next, bound the diameter $\Delta$ as follows
\beq \Delta^2 = 2.25\sup_{s,t \in T} \var\{\xi(s) - \xi(t)\} \le 2.25 \sup_{s,t\in T} 2 [\var\{\xi(s)\} + \var\{\xi(t)\}] \le 9 \kappa^2.\label{diam}\eeq
Now use $\log(1 + x) \le x$ to bound the right hand side of (\ref{entropy1}) by
$$9\int_0^{3\kappa} \sqrt{p\log\{1 + 1.5 c(b - a) / \epsilon\} }d\epsilon \le 9\sqrt{1.5 pc (b - a)} \int_0^{3\kappa} \epsilon^{-1/2} d\epsilon = B \kappa^{1/2}$$
as desired.
\item[(ii)]
Now, to calculate $\|\sup_{t \in S} |\xi(t)\|_\Psi$, note that the condition of Lemma 3.4 of \citet{pollard.mono} is trivially satisfied with $d(s, t) = \|\xi(s) - \xi(t)\|_\Psi = 1.5 [\var\{\xi(s) - \xi(t)\}]^{1/2}$ due to discreteness of $S$ and $D(\epsilon, T, d) \le |S|$ for all $\epsilon > 0$. Therefore we can apply (\ref{entropy1}) with $S$ instead of $T$ to conclude $\|\sup_{t \in S} |\xi(t)|\|_\Psi \le \Delta \sqrt{2 + \log |S|}$ where $\Delta^2 = \sup_{s,t \in S} d(s, t)^2 \le 9 \kappa_S^2$ as in (\ref{diam}).
\eenum
\end{proof}

\bibliographystyle{chicago}
\bibliography{TokdarReferences}

\begin{thebibliography}{}

\bibitem[\protect\citeauthoryear{Adler}{Adler}{1990}]{adler90}
Adler, R.~J. (1990).
\newblock {\em An introduction to continuity, extrema, and related topics for
  general Gaussian processes}, Volume~12.
\newblock Hayward, CA: Institute of Mathematical Statistics.

\bibitem[\protect\citeauthoryear{Banerjee, Gelfand, Finley, and Sang}{Banerjee
  et~al.}{2008}]{banerjee&etal08}
Banerjee, S., A.~E. Gelfand, A.~O. Finley, and H.~Sang (2008).
\newblock Gaussian predictive process models for large spatial data sets.
\newblock {\em Journal of the Royal Statistical Society Series B\/}~{\em
  70\/}(4), 825--848.

\bibitem[\protect\citeauthoryear{Besag}{Besag}{1974}]{besag74}
Besag, J.~E. (1974).
\newblock Spatial interaction and the statistical analysis of lattice systems.
\newblock {\em Journal of the Royal Statistical Society, Series B\/}~{\em 36},
  192--209.

\bibitem[\protect\citeauthoryear{Choi and Schervish}{Choi and
  Schervish}{2007}]{choi&schervish07}
Choi, T. and M.~Schervish (2007).
\newblock On posterior consistency in nonparametric regression problems.
\newblock {\em Journal of Multivariate Analysis\/}~{\em 98\/}(10), 1969--1987.

\bibitem[\protect\citeauthoryear{Csat\'{o}, Fokou\'{e}, Opper, Schottky, and
  Winther}{Csat\'{o} et~al.}{2000}]{csato.etal.00}
Csat\'{o}, L., E.~Fokou\'{e}, M.~Opper, B.~Schottky, and O.~Winther (2000).
\newblock Efficient approaches to gaussian process classification.
\newblock In S.~A. Solla, T.~K. Leen, and K.-R. M\"{u}ller (Eds.), {\em
  Advances in Neural Information Processing Systems}, Volume~12, Cambridge, MA.
  The MIT Press.

\bibitem[\protect\citeauthoryear{Finley, Sang, Banerjee, and Gelfand}{Finley
  et~al.}{2009}]{finley&etal09}
Finley, A.~O., H.~Sang, S.~Banerjee, and A.~E. Gelfand (2009).
\newblock Improving the performance of predictive process modeling for large
  datasets.
\newblock {\em Computational Statistics \& Data Analysis\/}~{\em 53\/}(8),
  2873--2884.

\bibitem[\protect\citeauthoryear{Foster, Waagen, Aijaz, Hurley, Luis, Rinsky,
  Satyavolu, Way, Gazis, and Srivastava}{Foster et~al.}{2009}]{foster&etal09}
Foster, L., A.~Waagen, N.~Aijaz, M.~Hurley, A.~Luis, J.~Rinsky, C.~Satyavolu,
  M.~J. Way, P.~Gazis, and A.~Srivastava (2009).
\newblock Stable and efficient gaussian process calculations.
\newblock {\em The Journal of Machine Learning Research\/}~{\em 10}, 857--882.

\bibitem[\protect\citeauthoryear{Ghosal and Roy}{Ghosal and
  Roy}{2006}]{ghosal&roy06}
Ghosal, S. and A.~Roy (2006).
\newblock Posterior consistency of gaussian process prior for nonparametric
  binary regression.
\newblock {\em The Annals of Statistics\/}~{\em 34}, 2413--2429.

\bibitem[\protect\citeauthoryear{Gilks, Richardson, and Spiegelhalter}{Gilks
  et~al.}{1995}]{mcmcbook2}
Gilks, W., S.~Richardson, and D.~Spiegelhalter (1995).
\newblock {\em Markov Chain Monte Carlo in Practice: Interdisciplinary
  Statistics}.
\newblock Chapman and Hall/CRC.

\bibitem[\protect\citeauthoryear{Guhaniyogi, Finley, Banerjee, and
  Gelfand}{Guhaniyogi et~al.}{2010}]{guhaniyogi.etal.11}
Guhaniyogi, R., A.~O. Finley, S.~Banerjee, and A.~E. Gelfand (2010).
\newblock Adaptive gaussian predictive process model for large spatial data
  sets.
\newblock American Statistical Association Joint Statistical Meeting. August 2,
  2010. Vancouver, British Columbia.

\bibitem[\protect\citeauthoryear{Handcock and Stein}{Handcock and
  Stein}{1993}]{handcock.stein.93}
Handcock, M.~S. and M.~L. Stein (1993).
\newblock A bayesian analysis of kriging.
\newblock {\em Technometrics\/}~{\em 35}, 403--410.

\bibitem[\protect\citeauthoryear{Kennedy and O'Hagan}{Kennedy and
  O'Hagan}{2001}]{kennedy&ohagan01}
Kennedy, M. and A.~O'Hagan (2001).
\newblock Bayesian calibration of computer models (with discussion).
\newblock {\em Journal of the Royal Statistical Society, Series b\/}~{\em 63},
  425--64.

\bibitem[\protect\citeauthoryear{Kim, Mallick, and Holmes}{Kim
  et~al.}{2005}]{kim.etal.05}
Kim, H.-M., B.~K. Mallick, and C.~C. Holmes (2005).
\newblock Analyzing nonstationary spatial data using piecewise gaussian
  processes.
\newblock {\em Journal of the American Statistical Association\/}~{\em 100},
  653--668.

\bibitem[\protect\citeauthoryear{Neal}{Neal}{1998}]{neal98}
Neal, R.~M. (1998).
\newblock Regression and classification using gaussian process priors.
\newblock In J.~M. Bernardo, J.~O. Berger, A.~P. Dawid, and A.~F.~M. Smith
  (Eds.), {\em Bayesian Statistics}, Volume~6, pp.\  475--501. Oxford
  University Press.

\bibitem[\protect\citeauthoryear{Oakley and OÕHagan}{Oakley and
  OÕHagan}{2002}]{oakley&ohagan02}
Oakley, J. and A.~OÕHagan (2002).
\newblock Bayesian inference for the uncertainty distribution of computer model
  outputs.
\newblock {\em Biometrika\/}~{\em 89}, 769--784.

\bibitem[\protect\citeauthoryear{Pace and Barry}{Pace and
  Barry}{1997}]{pace.barry.97}
Pace, R.~K. and R.~Barry (1997).
\newblock Quick computation of spatial autoregressive estimators.
\newblock {\em Geographical Analysis\/}~{\em 29}, 232--247.

\bibitem[\protect\citeauthoryear{Petrone, Guindani, and Gelfand}{Petrone
  et~al.}{2009}]{petrone&etal09}
Petrone, S., M.~Guindani, and A.~E. Gelfand (2009).
\newblock Hybrid dirichlet mixture models for functional data.
\newblock {\em Journal of the Royal Statistical Society: Series B (Statistical
  Methodology)\/}~{\em 71\/}(4), 755--782.

\bibitem[\protect\citeauthoryear{Pollard}{Pollard}{1990}]{pollard.mono}
Pollard, D. (1990).
\newblock {\em Empirical Processes: Theory and Applications}, Volume~2.
\newblock Institute of Mathematical Statistics and American Statistical
  Association.

\bibitem[\protect\citeauthoryear{Qui{\~n}onero-Candela and
  Rasmussen}{Qui{\~n}onero-Candela and Rasmussen}{2005}]{candela&williams05}
Qui{\~n}onero-Candela, J. and C.~E. Rasmussen (2005).
\newblock A unifying view of sparse approximate gaussian process regression.
\newblock {\em Journal of Machine Learning Research\/}~{\em 6}, 1939--1959.

\bibitem[\protect\citeauthoryear{Rasmussen and Williams}{Rasmussen and
  Williams}{2006}]{rasmussen&williams06}
Rasmussen, C.~E. and C.~K.~I. Williams (2006).
\newblock {\em Gaussian Processes for Machine Learning}.
\newblock The MIT Press.

\bibitem[\protect\citeauthoryear{Robert and Casella}{Robert and
  Casella}{2004}]{mcmcbook1}
Robert, C. and G.~Casella (2004).
\newblock {\em Monte Carlo Statistical Methods\/} (2 ed.).
\newblock Springer.

\bibitem[\protect\citeauthoryear{Rue, Martino, and Chopin}{Rue
  et~al.}{2009}]{rue.etal.09}
Rue, H., S.~Martino, and N.~Chopin (2009).
\newblock Approximate bayesian inference for latent gaussian models by using
  integrated nested laplace approximations.
\newblock {\em Journal of the Royal Statistical Society, Series B\/}~{\em 71},
  319--392.

\bibitem[\protect\citeauthoryear{Schwaighofer and Tresp}{Schwaighofer and
  Tresp}{2003}]{schwaighofer.tresp.03}
Schwaighofer, A. and V.~Tresp (2003).
\newblock Transductive and inductive methods for approximate gaussian process
  regression.
\newblock In S.~Becker, S.~Thrun, , and K.~Obermayer (Eds.), {\em Advances in
  Neural Information Processing Systems}, Volume~15, Cambridge, Massachussetts,
  pp.\  953--960. The MIT Press.

\bibitem[\protect\citeauthoryear{Seeger, Williams, and Lawrence}{Seeger
  et~al.}{2003}]{seeger.etal.03}
Seeger, M., C.~K.~I. Williams, and N.~Lawrence (2003).
\newblock Fast forward selection to speed up sparse gaussian process
  regression.
\newblock In C.~M. Bishop and B.~J. Frey (Eds.), {\em Ninth International
  Workshop on Artificial Intelligence and Statistics}. Society for Artificial
  Intelligence and Statistics.

\bibitem[\protect\citeauthoryear{Shi and Wang}{Shi and Wang}{2008}]{shi&wang08}
Shi, J.~Q. and B.~Wang (2008).
\newblock Curve prediction and clustering with mixtures of gaussian process
  functional regression models.
\newblock {\em Statistical Computing\/}~{\em 18}, 267--283.

\bibitem[\protect\citeauthoryear{Short, Higdon, and Kronberg}{Short
  et~al.}{2007}]{short.etal.07}
Short, M.~B., D.~M. Higdon, and P.~P. Kronberg (2007).
\newblock Estimation of faraday rotation measures of the near galactic sky
  using gaussian process models.
\newblock {\em Bayesian Analysis\/}~{\em 2}, 665--680.

\bibitem[\protect\citeauthoryear{Smola and Bartlett}{Smola and
  Bartlett}{2001}]{smola.bartlett.01}
Smola, A.~J. and P.~L. Bartlett (2001).
\newblock Sparse greedy gaussian process regression.
\newblock In {\em Advances in Neural Information Processing Systems},
  Volume~13, Cambridge, Massachussetts, pp.\  619--625. The MIT Press.

\bibitem[\protect\citeauthoryear{Snelson and Ghahramani}{Snelson and
  Ghahramani}{2006}]{snelson&ghahramani06}
Snelson, E. and Z.~Ghahramani (2006).
\newblock Sparse gaussian processes using pseudo-inputs.
\newblock In Y.~Weiss, B.~Sch\"{o}lkopf, and J.~Platt (Eds.), {\em Advances in
  Neural Information Processing Systems}, Volume~18, Cambrisge, Massachussetts.
  The MIT Press.

\bibitem[\protect\citeauthoryear{Sudderth and Jordan}{Sudderth and
  Jordan}{2009}]{sudderth&jordan09}
Sudderth, E. and M.~Jordan (2009).
\newblock Shared segmentation of natural scenes using dependent pitman-yor
  processes.
\newblock In D.~Koller, D.~Schuurmans, Y.~Bengio, and L.~Bottou (Eds.), {\em
  Advances in Neural Information Processing Systems 21}. MIT Press.

\bibitem[\protect\citeauthoryear{Tokdar}{Tokdar}{2007}]{tokdar07}
Tokdar, S.~T. (2007).
\newblock Towards a faster implementation of density estimation with logistic
  gaussian process priors.
\newblock {\em Journal of Computational and Graphical Statistics\/}~{\em 16},
  633--655.

\bibitem[\protect\citeauthoryear{Tokdar and Ghosh}{Tokdar and
  Ghosh}{2007}]{tokdar&ghosh07}
Tokdar, S.~T. and J.~K. Ghosh (2007).
\newblock Posterior consistency of logistic gaussian process priors in density
  estimation.
\newblock {\em Journal of Statistical Planning and Inference\/}~{\em 137},
  34--42.

\bibitem[\protect\citeauthoryear{Tokdar and Kadane}{Tokdar and
  Kadane}{2011}]{tokdar&kadane11}
Tokdar, S.~T. and J.~B. Kadane (2011).
\newblock Simultaneous linear quantile regression: A semiparametric bayesian
  approach.
\newblock Duke Statistical Science Discussion Paper \#12.

\bibitem[\protect\citeauthoryear{Tokdar, Zhu, and Ghosh}{Tokdar
  et~al.}{2010}]{tokdar&etal10}
Tokdar, S.~T., Y.~M. Zhu, and J.~K. Ghosh (2010).
\newblock Density regression with logistic gaussian process priors and subspace
  projection.
\newblock {\em {B}ayesian Analayis\/}~{\em 5\/}(2), 316--344.

\bibitem[\protect\citeauthoryear{van~der Vaart and van Zanten}{van~der Vaart
  and van Zanten}{2008}]{vandervaart&vanzanten08}
van~der Vaart, A.~W. and J.~H. van Zanten (2008).
\newblock Rates of contraction of posterior distributions based on gaussian
  process priors.
\newblock {\em Annals of Statistics\/}~{\em 36}, 1435--1463.

\bibitem[\protect\citeauthoryear{van~der Vaart and van Zanten}{van~der Vaart
  and van Zanten}{2009}]{vandervaart&vanzanten09}
van~der Vaart, A.~W. and J.~H. van Zanten (2009).
\newblock Adaptive bayesian estimation using a gaussian random field with
  inverse gamma bandwidth.
\newblock {\em The Annal of Statistics\/}~{\em 37\/}(5B), 2655--2675.

\bibitem[\protect\citeauthoryear{Zhu and Stein}{Zhu and
  Stein}{2005}]{zhu.stein.05}
Zhu, Z. and M.~Stein (2005).
\newblock Spatial sampling design for parameter estimation of the covariance
  function.
\newblock {\em Journal of Statistical Planning and Inference\/}~{\em 134},
  583--603.

\bibitem[\protect\citeauthoryear{Zimmerman}{Zimmerman}{2006}]{zimmerman.06}
Zimmerman, D. (2006).
\newblock Optimal network design for spatial prediction, covariance parameter
  estimation, and empirical prediction.
\newblock {\em Environmetrics\/}~{\em 17}, 635--652.

\end{thebibliography}

\end{document}